\documentclass[journal]{IEEEtran}

\usepackage{amssymb,amsthm}
\usepackage[centertags]{amsmath}
\usepackage{graphicx}
\usepackage{booktabs,multirow}
\usepackage[numbers,sort&compress]{natbib}
\usepackage{url}
\usepackage{xcolor}
\usepackage{float}

\usepackage[ruled,vlined,linesnumbered]{algorithm2e}
\SetAlgoCaptionSeparator{:}
\SetAlgoNlRelativeSize{-1}
\SetAlFnt{\small}
\SetAlCapFnt{\small}
\SetAlCapNameFnt{\small\bfseries}
\SetKwSty{textbf}
\SetKw{Return}{return}
\SetKwFunction{Segment}{Segment}
\SetKwProg{Fn}{Procedure}{:}{end}

\usepackage[colorlinks=true,allcolors=blue]{hyperref}

\theoremstyle{plain}

\newtheorem{theorem}{Theorem}
\newtheorem{proposition}{Proposition}
\newtheorem{lemma}{Lemma}

\theoremstyle{definition}

\newtheorem{definition}{Definition}
\newtheorem{assumption}{Assumption}

\begin{document}

\title{Spectral Dynamics of DeepWalk Embeddings for Dynamic Network Change-Point Detection}

\author{Houlin Zhou, Yejin Wang, Xufei Tang, and Dan Zhuang%
\thanks{Houlin Zhou and Yejin Wang contributed equally to this work.
Corresponding author: Dan Zhuang.}%
\thanks{Houlin Zhou is with the Department of Statistics, Faculty of Arts and Sciences,
Beijing Normal University, Zhuhai, Guangdong 519085, China
(e-mail: \mbox{houlinzhou@bnu.edu.cn}).}%
\thanks{Yejin Wang is with the School of Big Data and Statistics,
Anhui University, Hefei 230601, China. (e-mail: \mbox{w326112006@stu.ahu.edu.cn}).}%
\thanks{Xufei Tang is with the School of Mathematics and Statistics,
Hefei Normal University, Hefei 230601, China
(e-mail: \mbox{xufei\_tang@163.com}).}%
\thanks{Dan Zhuang is with the School of Mathematics and Statistics,
Fujian Normal University, Fuzhou 350117, China
(e-mail: \mbox{zhgdan@fjnu.edu.cn}).}%
}

\maketitle

\begin{abstract}\mdseries
Dynamic networks describe evolving relational systems in which abrupt structural changes may signal anomalous events or important transitions. Detecting such changes requires distinguishing genuine structural signals from fluctuations in network observations and learned representations. We propose a DeepWalk-based framework for detecting and localizing structural changes in dynamic networks. For each snapshot, we learn low-dimensional node embeddings, align them to a fixed reference using orthogonal Procrustes transformations, and aggregate them by mean pooling into comparable graph-level vectors. We then construct a multivariate cumulative sum (CUSUM) scan statistic to identify changes in the embedding mean. Under stated regularity conditions, we derive bounds on stochastic fluctuations under the null hypothesis and sufficient conditions for reliable detection and consistent localization under the alternative. The analysis connects detection performance with network sparsity, spectral separation, embedding dimension, and change magnitude, clarifying the balance between structural signals and representation variability.
Simulation studies demonstrate the effectiveness of the proposed method in detecting and localizing structural changes across a range of dynamic network settings. An application to an international trade network further illustrates its practical usefulness in identifying shifts in trade allocation.

\end{abstract}


\section{Introduction}
\label{sec:introduction}

Networks provide a natural representation for systems composed of interacting entities and have been widely applied in various domains, including social networks, biological systems, finance, communication systems, and information systems
\citep{guyon1995random, newman2011structure, nagurney2012financial, kolaczyk2014statistical, albert2002statistical, kim2018review, han2021dynamic, bianchi2024relational}.
In many real-world applications, the underlying networks are not static but continuously evolve over time.
Driven by external events or internal structural transitions, network connectivity patterns, community organizations, and interaction mechanisms may change dynamically.
Recently, continuous-time dynamic network learning approaches have been developed to jointly characterize the temporal evolution of network topology and interaction processes
\citep{10636764}.
Therefore, identifying when such changes occur is fundamental for understanding the evolution of network systems and discovering anomalous events or scientifically meaningful phenomena.

Change-point detection in dynamic networks is substantially more challenging than that in classical Euclidean-valued sequences.
For conventional vector-valued observations, structural changes can often be characterized by variations in certain statistical quantities, such as means, variances, or regression coefficients.
In contrast, each observation in a dynamic network is itself a high-dimensional and non-Euclidean object with complex dependencies and structural heterogeneity.
Network changes may appear as variations in edge probabilities, community memberships, latent positions, connectivity patterns, or other high-order graph structural characteristics.
Therefore, comparing network observations at different time points requires either explicitly specifying a statistical model for the graph generation mechanism or constructing an informative representation capable of effectively measuring structural changes.

Existing approaches for change-point detection in dynamic networks can be broadly categorized into three classes: model-based methods, similarity-based methods, and embedding-based methods.
Model-based methods characterize structural changes through variations in parameters of probabilistic network models, including stochastic block models, nonhomogeneous Bernoulli networks, and random dot product graphs
\citep{bhattacharjee2020change, 10.1214/20-AOS1953, enikeeva2021change, padilla2022change}.
Similarity-based methods compare network observations at different time points using manually designed or data-driven discrepancy measures
\citep{10.1214/14-AOS1269, chen2019sequential, penaloza2024changepoint, sulem2024graph}.
Embedding-based methods first map each network into a low-dimensional representation and then detect temporal changes from the resulting embedding sequence
\citep{zambon2019changepoint,lin2023network,huang2020laplacian, hewapathirana2020change, liu2024identifying, huang2022learning, xu2022dyng2g}.
Although these approaches have significantly advanced dynamic network
analysis, several important limitations remain. Model-based methods rely
on the validity of assumed graph generation mechanisms, whereas
similarity-based methods are sensitive to the selected discrepancy
measures. General graph-to-vector change-point inference has received
theoretical treatment \citep{zambon2019changepoint}, but the end-to-end
stochastic effects of random-walk sampling, DeepWalk estimation, and
cross-snapshot coordinate alignment on downstream scan statistics remain
insufficiently understood. A more detailed review of these research
directions is provided in Section~\ref{sec:related_work}.

Among these approaches, graph embedding provides an attractive balance between modeling flexibility and statistical tractability.
By representing irregular graph objects in a low-dimensional Euclidean space, graph embedding enables the application of classical statistical tools without requiring a fully specified distributional model for the network sequence.
Existing studies on joint embedding of multiple graphs have demonstrated that extracting comparable low-dimensional graph-level features from a collection of networks can effectively support graph comparison, classification, and statistical inference
\citep{8889404}.
More recently, continuous-time dynamic network representation methods have further utilized neural differential equations to characterize the joint evolution of network topology and node states, demonstrating that learning-based embeddings can extract low-dimensional representations that evolve over time from complex dynamic networks
\citep{11192661}.
In particular, DeepWalk
\citep{perozzi2014deepwalk}
learns node representations through local co-occurrence patterns generated by truncated random walks and has achieved strong empirical performance in graph representation learning.
However, directly applying change-point detection procedures to DeepWalk embeddings raises a fundamental question:
how can one distinguish genuine structural transitions in the underlying network from random fluctuations introduced by network sampling, random-walk generation, and embedding estimation?

This question reveals a critical theoretical gap in existing embedding-based change-point detection methods.
A large scan statistic may indicate a structural transition in the network, but it may also result from intrinsic randomness in network observations or the embedding procedure.
It remains unclear how network size, temporal sample size, embedding dimension, graph sparsity, and the magnitude of structural changes jointly determine the behavior of change-point statistics constructed from graph embeddings.
Without corresponding theoretical analysis, it is difficult to obtain a principled criterion for determining whether an observed change is sufficiently strong to be distinguished from random fluctuations under the null hypothesis.
Consequently, the statistical reliability of embedding-based change-point detection for networks remains insufficiently understood.

The key observation of this paper is that, once the random errors introduced by network observations and embedding procedures can be explicitly controlled, the resulting graph embedding sequence can be regarded as a sequence of structured Euclidean representations.
This perspective establishes a connection between graph representation learning and statistical change-point analysis.
Rather than treating DeepWalk merely as a heuristic feature extraction technique, we further investigate how uncertainty in the learned representations propagates to downstream CUSUM-type scan statistics.
The resulting theoretical analysis separates the deterministic signal induced by structural transitions from the random fluctuations caused by network generation and embedding procedures.

Based on this observation, we propose a DeepWalk-based framework for change-point detection in dynamic networks.
The proposed framework first maps each network snapshot into low-dimensional node representations, aligns independently fitted embeddings to a fixed reference coordinate system, and obtains graph-level vectors through a permutation-invariant aggregation operator.
Then, a multivariate CUSUM scan statistic is constructed on the sequence of graph-level representations to detect and localize potential structural changes.
Unlike approaches that simply regard embeddings as empirical features, we further analyze how errors from network sampling, random walks, embedding estimation, and coordinate alignment propagate into the scan statistic, thereby characterizing the statistical reliability of the proposed method.
The complete methodological framework is presented in Section~\ref{sec:method}, and the theoretical results are provided in Section~\ref{sec:theory}.

The main contributions of this paper are summarized as follows.

\begin{itemize}

\item 
We propose a representation-based change-point detection framework for dynamic networks.
The proposed framework utilizes DeepWalk to map non-Euclidean network snapshots into a low-dimensional Euclidean representation space, uses fixed-reference orthogonal Procrustes transformations to standardize the coordinate orientation of all snapshots, and constructs a CUSUM-type scan statistic within this unified representation space.
Unlike methods relying on parametric network models or prespecified graph discrepancy measures, the proposed approach reduces dependence on specific generative mechanisms and manually designed similarity measures.

\item
We establish a theoretical analysis framework connecting graph representation learning with statistical change-point detection.
By characterizing the uncertainty propagation process of DeepWalk embeddings and their alignment, we derive the stochastic fluctuation bound of the statistic under no structural change, the detection condition under structural transitions, and the consistency guarantee for change-point localization.
Furthermore, we reveal the relationship among network structures, representation quality, and detection performance.

\item
We validate the effectiveness of the proposed method through extensive simulation studies and real dynamic network applications.
The experiments cover various network evolution mechanisms, structural change magnitudes, and network scales.
Furthermore, we apply the proposed method to a real-world dynamic trade network to analyze structural changes and demonstrate its practical applicability.

\end{itemize}

\subsection{Related Work}
\label{sec:related_work}

\textit{Model-based Change-point Detection in Dynamic Networks:}
An important line of research studies change-point detection in dynamic networks by constructing explicit probabilistic models for the observed graph sequences.
Classical statistical network models include Gaussian graphical models
\citep{dempster1972covariance},
stochastic block models
\citep{holland1983stochastic},
and more general probabilistic graphical models
\citep{lauritzen1996graphical}.
These models provide interpretable descriptions of dependency structures, community organizations, and latent connectivity patterns, and have served as statistical foundations for many dynamic network analysis methods.

For dynamic networks, model-based change-point detection methods typically identify change points by detecting whether the parameters governing edge formation mechanisms or latent graph structures have changed.
For example, existing studies have investigated change-point detection under various frameworks, including dynamic stochastic block models
\citep{bhattacharjee2020change},
sparse nonhomogeneous Bernoulli network models
\citep{10.1214/20-AOS1953},
dynamic stochastic block models with missing connections
\citep{enikeeva2021change},
and dynamic nonparametric random dot product graphs with dependent structures
\citep{padilla2022change}.
An important advantage of these approaches is that structural changes can usually be directly associated with changes in specific interpretable model parameters, such as block connection probabilities, community memberships, or latent positions.

However, the theoretical guarantees of these methods typically rely on the correctness of the assumed graph generation models.
For complex real-world networks, selecting and validating an appropriate probabilistic model can be difficult.
Moreover, when networks are large-scale, sparse, heterogeneous, temporally dependent, or partially observed, parameter estimation may introduce substantial computational costs.
Therefore, this paper does not require the graph sequence to follow a fully specified dynamic stochastic block model, Bernoulli network model, or latent position model.
Instead, we first use DeepWalk to summarize the random-walk connectivity patterns within each network snapshot and then detect changes from the resulting graph-level representation sequence.
Such a representation-based approach sacrifices some direct interpretability at the model-parameter level, but reduces dependence on the correctness of a specific generative mechanism.

\textit{Similarity-based Change-point Detection:}
Another line of research detects network changes by measuring discrepancies between observations before and after candidate time points.
This perspective is closely related to nonparametric change-point detection, where the goal is to identify distributional changes without fully specifying a probabilistic model.

Existing studies have proposed graph-structure-based and nearest-neighbor-based methods for detecting general distributional changes
\citep{10.1214/14-AOS1269,chen2019sequential}.
More recently, several approaches have further constructed data-driven discrepancy measures using learned graph representations or neural network architectures.
Representative methods include graph neural network-based change-point detection approaches
\citep{penaloza2024changepoint}
and graph convolutional network-based methods
\citep{sulem2024graph}.
These approaches are capable of capturing nonlinear differences and high-order structural variations that may not be reflected by direct edge-wise comparisons.

Similarity-based methods provide substantial flexibility because they do not require direct estimation of parametric network models.
However, their effectiveness strongly depends on whether the selected similarity or discrepancy measure is sensitive to the specific types of structural changes of interest.
For example, a measure designed primarily for local edge perturbations may fail to detect community restructuring, while a global graph statistic may obscure local structural changes.
Furthermore, when similarity measures are learned from data, their statistical randomness is often difficult to quantify accurately.
Different from these approaches, our method does not prespecify a distance between graphs or directly use learned similarity scores as detection statistics.
Instead, we first map all network snapshots into a common Euclidean representation space and then apply a CUSUM statistic with a clear interpretation in terms of mean comparison.
Consequently, graph comparison is transformed into a statistical inference problem in a unified representation space.

\textit{Graph Embedding and Embedding-based Change Detection:}
Graph embedding aims to map nodes or entire graphs into low-dimensional spaces while preserving informative structural properties.
Traditional dimensionality reduction techniques include principal component analysis
\citep{pearson1901liii},
principal curves
\citep{hastie1984principal,hastie1989principal},
locally linear embedding
\citep{roweis2000nonlinear},
and ISOMAP
\citep{tenenbaum2000global}.
Although these methods were not originally designed for graph data, they established a fundamental idea that complex high-dimensional observations can often be studied through low-dimensional representations.

For network data, spectral methods typically construct embeddings based on graph matrices such as adjacency matrices or graph Laplacian matrices.
Accordingly, existing dynamic network change-point detection methods have utilized Laplacian spectral embedding
\citep{huang2020laplacian,huang2024laplacian},
singular value decomposition
\citep{hewapathirana2020change},
and truncated singular value decomposition
\citep{liu2024identifying}
to construct network representations.
Other studies represent dynamic networks on constant-curvature Riemannian manifolds
\citep{8782142},
allowing latent representations to possess non-Euclidean geometric structures.

Deep graph representation learning provides another approach for
extracting structural information from networks.
DeepWalk
\citep{perozzi2014deepwalk}
learns node representations through co-occurrence relationships generated
by truncated random walks.
Qiu et al.~\citep{qiu2018network} showed that DeepWalk and several
related skip-gram-based network embedding methods trained with negative
sampling admit a unified matrix-factorization interpretation, thereby
connecting random-walk co-occurrence statistics with the spectral
structure of the underlying graph.
Graph convolutional networks
\citep{kipf2016semi}
aggregate neighborhood information through learned transformations.
Subsequently, various representation learning methods have been applied
to dynamic network analysis, including doc2vec-based representations
\citep{huang2022learning},
DynG2G
\citep{xu2022dyng2g},
VGGM
\citep{zhang2024vggm},
and TransformerG2G
\citep{varghese2024transformerg2g}.
These studies demonstrate that learned embeddings can effectively
summarize complex temporal and structural information in dynamic
networks.

Most closely related to our change-point setting, Lin et
al.~\citep{lin2023network} first embed network nodes into a latent space
and then construct a network descriptor using random walks over the
latent node positions. In contrast, our random walks generate the
co-occurrence statistics used to learn DeepWalk embeddings. We further
align the snapshot-wise embeddings, aggregate them into graph-level
representations, and derive explicit fluctuation, detection, and
localization guarantees for the resulting CUSUM statistic.

Embedding-based methods are attractive because they transform graph observations into Euclidean objects, thereby enabling the application of classical statistical techniques.
However, existing studies mainly focus on algorithmic design and empirical performance.
The stochastic behavior of change-point statistics constructed from learned embeddings remains insufficiently understood.
In particular, how to distinguish genuine signals caused by network structural transitions from errors induced by graph randomness, random-walk sampling, and embedding estimation has received limited attention.
Different from existing embedding-based approaches that mainly rely on empirical performance, this paper not only utilizes DeepWalk representations for network snapshots but also explicitly characterizes how uncertainties from network sampling, random walks, and embedding estimation propagate into downstream CUSUM statistics.
The resulting signal-to-noise separation condition explains when representation-based dynamic network change-point detection can achieve statistical reliability.

\subsection{Paper Organization and Notation}

The remainder of this paper is organized as follows.
Section~\ref{sec:problem} introduces the dynamic network change-point detection problem considered in this paper, including the dynamic network generation setting, the DeepWalk graph representation learning procedure, and the CUSUM detection framework based on graph embedding sequences.
Section~\ref{sec:theory} establishes the theoretical properties of the proposed method.
We first analyze the asymptotic behavior of the detection statistic under the hypothesis of no structural change and then investigate the signal enhancement property of the statistic under the presence of a single structural change, leading to the consistency of change-point estimation.
Section~\ref{sec:multiple} extends the proposed method to scenarios with multiple structural changes, introduces a recursive binary segmentation procedure based on local CUSUM statistics, and establishes theoretical guarantees for multiple change-point detection and localization.
Section~\ref{sec:simulation} evaluates the detection performance of the proposed method through systematic simulation studies under different dynamic network evolution mechanisms and structural change patterns.
Section~\ref{sec:real} further applies the proposed method to real dynamic network data to demonstrate its effectiveness in practical structural change analysis tasks.
Finally, Section~\ref{sec:conclusion} concludes the paper and discusses future research directions.
Proofs of the main theoretical results are provided in the Appendix.

Before proceeding, we introduce several notations used throughout this paper.
We use $\|\cdot\|_2$ to denote the Euclidean norm of a vector,
$\|\cdot\|_{\mathrm{op}}$ to denote the spectral norm of a matrix,
and $\|\cdot\|_F$ to denote the Frobenius norm of a matrix.
When there is no ambiguity, the vector norm $\|\cdot\|_2$ is abbreviated as $\|\cdot\|$.
For a positive integer $n$, let
$\mathbf 1=(1,\ldots,1)^\top\in\mathbb R^n$
denote the all-one vector,
and let
$\mathbb S^{d-1}
=\{u\in\mathbb R^d:\|u\|_2=1\}$
denote the unit sphere in $\mathbb R^d$.
The quantity $\rho=\rho_n\in(0,1]$ denotes the network sparsity parameter, such that $n\rho$ characterizes the typical expected-degree scale of the network. For a sequence of random variables $X_T$ and a sequence of positive numbers $a_T$,
$X_T=O_P(a_T)$ denotes that
$X_T/a_T$ is bounded in probability;
$a_T\gg b_T$ denotes that
$b_T/a_T\to0$.
The notation $\overset{P}{\longrightarrow}$ denotes convergence in probability.
For matrices $U,V\in\mathbb R^{n\times K}$,
if their columns form two orthonormal bases of two $K$-dimensional subspaces, respectively,
then $\Theta(U,V)$ denotes the matrix of principal angles between the two subspaces,
and
$\|\sin\Theta(U,V)\|_{\mathrm{op}}$
denotes the sine of the largest principal angle.

\section{Problem Formulation and Methodology}
\label{sec:problem}
\label{sec:method}

\subsection{Dynamic Networks and Change-point Detection}

This paper considers a dynamic network sequence consisting of a series of network snapshots observed sequentially over time,
$\mathcal{G}_{1:T}=\{G_t\}_{t=1}^{T},$
where the \(t\)-th network snapshot is represented as
$G_t=(V,E_t),$
with \(V\) denoting the fixed node set and \(E_t\) denoting the edge set at time \(t\). Each \(G_t\) is an undirected graph, and the number of nodes satisfies \(|V|=n\).
The dynamic network sequence $\mathcal G_{1:T}$ describes the evolution of relational structures among the same set of entities over time, where each network snapshot characterizes the node interaction relationships at a specific time point.

Let $\mathcal{G}_n$ denote the graph space consisting of all undirected graphs with $n$ nodes.
We regard each dynamic network snapshot as a random object defined on the graph space $\mathcal{G}_n$ and assume that
$G_t\sim\mathbb P_t,\quad t=1,\ldots,T$,
where $\mathbb P_t$ denotes the network generation distribution at time $t$.
It is important to emphasize that although the observed network snapshot $G_t$ is a deterministic graph structure, the above probabilistic representation does not imply that an individual observed graph itself is random.
Instead, it is introduced to characterize the intrinsic uncertainty in the network generation process.
This statistical modeling perspective enables the structural evolution of dynamic networks to be described through time-varying probability distributions $\mathbb P_t$.
In particular, changes in $\mathbb P_t$ may correspond to variations in network structural characteristics, such as community structures, node connectivity patterns, latent representation spaces, or other high-level network properties.

\begin{definition}[Structural Change in Dynamic Networks]
\label{def:network-change}
If there exists $k^*\in\{1,\ldots,T-1\}$ such that
\begin{equation*}
\mathbb P_1=\cdots=\mathbb P_{k^*}\neq\mathbb P_{k^*+1}=\cdots=\mathbb P_T,
\end{equation*}
then a structural change occurs at time $k^*$, which is referred to as a change point.
\end{definition}

Definition~\ref{def:network-change} does not require $\mathbb P_t$ to belong to any specific network model, and therefore covers stochastic block models, random dot product graphs, graphons, and more general latent space models.
Since direct statistical inference on the graph space is challenging, we consider a representation map that explicitly includes embedding, coordinate alignment, and graph-level aggregation,
$\phi:\mathcal G_n\to\mathbb R^d$,
and define $Z_t=\phi(G_t;X_{\rm ref})$, where $X_{\rm ref}$ is a reference embedding held fixed across time.
If the embedding preserves the dominant structures associated with the network generation mechanism, changes in $\mathbb P_t$ will induce changes in the distribution or moments of $Z_t$, thereby transforming the graph-valued change-point problem into a multivariate change-point problem in Euclidean space.
Therefore, the hypothesis testing problem considered in this paper is
\begin{align}
	H_0:\quad
	& \mathcal L(Z_1)=\cdots=\mathcal L(Z_T),
	\notag\\
	H_1:\quad
	& \exists\,k^*\in\{1,\ldots,T-1\}\ \text{such that}
	\label{eq:hypothesis}\\
	& \mathcal L(Z_1)=\cdots=\mathcal L(Z_{k^*})
	\neq
	\mathcal L(Z_{k^*+1})=\cdots=\mathcal L(Z_T).
	\notag
\end{align}

\subsection{DeepWalk Embedding Representation}

Given a static undirected graph $G=(V,E)$, where $V=\{1,\ldots,n\}$,
let $A=(A_{ij})_{1\leq i,j\leq n}$ denote the adjacency matrix of graph $G$, where
$A_{ij}=1$ indicates that an edge exists between nodes $i$ and $j$, and
$A_{ij}=0$ otherwise.
For node $i$, define its degree as
$d_i=\sum_{j=1}^{n}A_{ij}$,
and let the degree matrix be
$D_A=\operatorname{diag}(d_1,\ldots,d_n)$.
The corresponding random-walk transition probability matrix is defined as
$P=D_A^{-1}A$,
where $P_{ij}$ represents the probability that a random walk moves from node $i$ to node $j$.
The main method and theory assume that all nodes used in an embedding have positive degree. If a snapshot contains isolated nodes, they are excluded from that snapshot's walks, and common active nodes are subsequently used as anchors for temporal alignment.
Suppose that each node generates $r$ random walks of length $L$.
Let
$w^{(m)}=(w_1^{(m)},\ldots,w_L^{(m)})$
denote the $m$-th walk sequence, and let $T_w$ denote the context window size.

Based on the local co-occurrence relationships in random walks, we define a symmetric co-occurrence matrix $C=(C_{ij})$, where

\begin{align*}
C_{ij}={}&\sum_{h=1}^{T_w}\sum_{m=1}^{r}\sum_{\ell=1}^{L-h}
\mathbf 1\{w_\ell^{(m)}=i,w_{\ell+h}^{(m)}=j\}\\
&+\sum_{h=1}^{T_w}\sum_{m=1}^{r}\sum_{\ell=1}^{L-h}
\mathbf 1\{w_\ell^{(m)}=j,w_{\ell+h}^{(m)}=i\}.
\end{align*}

$C_{ij}$ describes the co-occurrence frequency of nodes $i$ and $j$ within random-walk windows, and therefore captures both local connectivity patterns and high-order structural information induced by multi-step random walks.

Let
$X=(X_1,\ldots,X_n)^\top\in\mathbb R^{n\times d}$
and
$Y=(Y_1,\ldots,Y_n)^\top\in\mathbb R^{n\times d}$
denote the node embedding matrix and context embedding matrix, respectively.
Here,
$X_i,Y_i\in\mathbb R^d$
represent the low-dimensional latent representations of node $i$,
and $d$ denotes the embedding dimension.
DeepWalk learns low-dimensional representations by minimizing the negative log-likelihood

\begin{align}
\mathcal L(C;X,Y)&=-\sum_{i=1}^n\sum_{j=1}^n C_{ij}\log Q_{ij},\notag\\
Q_{ij}&=\frac{\exp(\langle X_i,Y_j\rangle)}
{\sum_{k=1}^n\exp(\langle X_i,Y_k\rangle)}.
\label{eq:deepwalk-loss}
\end{align}

Applying the above procedure independently to each snapshot $G_t$ yields a node embedding matrix
$X_t=(X_{t,1},\ldots,X_{t,n})^\top\in\mathbb R^{n\times d}$.
Because DeepWalk involves random initialization, random-walk sampling, and non-convex optimization, independently fitted embeddings do not have naturally matched coordinate axes. Even when two snapshots have the same structure, their learned node configurations may differ by a global rotation or reflection. Directly pooling and comparing the unaligned $X_t$ may therefore turn an algorithmic change of orientation into a spurious network change.

We remove this ambiguity by aligning every embedding to a reference
$X_{\rm ref}\in\mathbb R^{n\times d}$ that is fixed across time. In the theoretical analysis, $X_{\rm ref}$ is treated as fixed or estimated from an independent reference sample. In implementation, when no external reference is available, we use $X_{\rm ref}=X_1$ and assess sensitivity to the reference choice. For every $t$, define the orthogonal Procrustes estimator
\begin{equation}
	\begin{aligned}
		\widehat Q_t
		&=
		\arg\min_{Q\in\mathcal O(d)}
		\|X_tQ-X_{\rm ref}\|_F^2,\\
		\mathcal O(d)
		&=
		\left\{
		Q\in\mathbb R^{d\times d}:Q^\top Q=I_d
		\right\}.
	\end{aligned}
	\label{eq:procrustes}
\end{equation}
If $X_t^\top X_{\rm ref}=U_t\Sigma_tV_t^\top$ is a singular value decomposition, one solution is $\widehat Q_t=U_tV_t^\top$. The aligned node embedding is
\begin{equation}
\widetilde X_t=X_t\widehat Q_t.
\label{eq:aligned-node-embedding}
\end{equation}
Orthogonal transformations preserve inner products and Euclidean distances, so alignment standardizes only the overall coordinate orientation without changing the within-snapshot embedding geometry.

When the active node set varies with time, let $V_t^+$ and $V_{\rm ref}^+$ denote the active nodes in snapshot $t$ and in the reference snapshot, and define the common anchors $\mathcal A_t=V_t^+\cap V_{\rm ref}^+$. We then solve
\begin{equation}
\widehat Q_t
=
\arg\min_{Q\in\mathcal O(d)}
\|X_{t,\mathcal A_t}Q-X_{{\rm ref},\mathcal A_t}\|_F^2
\label{eq:anchor-procrustes}
\end{equation}
and apply the resulting transformation to all active nodes in snapshot $t$. This extension requires sufficiently many informative anchors to identify the $d$ coordinate directions.

The Procrustes step should be interpreted as a coordinate normalization associated with the dominant spectral-subspace formulation used in our theory. We do not claim that an unrestricted two-matrix Skip-Gram parameterization is identifiable only up to orthogonal transformations. Instead, the theoretical node representation is viewed as an empirical basis of the dominant DeepWalk operator subspace, whose rotation and reflection ambiguity is removed by~\eqref{eq:procrustes}.

After alignment, we construct a graph-level representation by mean pooling,
\begin{equation}
Z_t
=\frac1n\sum_{i=1}^n\widetilde X_{t,i}
=\frac1n\mathbf 1_n^\top\widetilde X_t
\in\mathbb R^d.
\label{eq:graph-embedding}
\end{equation}
Thus, the dynamic graph sequence is converted into a low-dimensional multivariate sequence $\{Z_t\}_{t=1}^T$ expressed in a common reference coordinate system. Mean pooling is permutation invariant after node correspondence has been used for alignment, and it provides a simple first-moment summary of the aligned node-embedding distribution. Other permutation-invariant graph-level aggregators may replace it when changes beyond the first moment are of interest.

\subsection{Change-point Model Based on Embedding Mean}

Based on the graph-level embedding sequence
$\{Z_t\}_{t=1}^{T}$
obtained in the previous section, we further transform the dynamic network change-point detection problem into a mean change detection problem in the embedding space.
Let
$\mu_t=\mathbb E(Z_t)\in\mathbb R^d$
denote the mean of the graph embedding corresponding to the $t$-th network snapshot.
Since graph embeddings are learned by accumulating network structural information through random walks, changes in the dynamic network generation distribution $\mathbb P_t$ may induce corresponding changes in the embedding distribution and its statistical moments.
Therefore, this paper identifies potential structural transitions in dynamic networks by detecting changes in the embedding mean.

We first consider the single change-point case, where there exists an unknown change point
$k^*\in\{1,\ldots,T-1\}$,
such that
$\mu_1=\cdots=\mu_{k^*}
\neq
\mu_{k^*+1}
=\cdots
=\mu_T .$

The corresponding hypothesis testing problem is
\begin{align*}
H_0:&\quad \mu_1=\cdots=\mu_T,\notag\\
H_1:&\quad \exists k^*
\ \text{such that}\ 
\mu_1=\cdots=\mu_{k^*}
\neq
\mu_{k^*+1}
=\cdots
=\mu_T .
\end{align*}

This formulation does not require explicit estimation of community labels or specification of particular network generation model parameters.
Instead, it characterizes network structural evolution indirectly through changes in statistical moments of graph embeddings.
Although this paper mainly focuses on changes in the first-order moment of embeddings, the framework can be naturally extended to more general distributional changes, including changes in embedding covariance structures, kernel mean embeddings, and distributional shifts measured by Wasserstein distances.

\subsection{CUSUM Statistic and Change-point Estimation}

Based on the embedding sequence
$\{Z_t\}_{t=1}^{T}$,
for any candidate change-point location
$k\in\{1,\ldots,T-1\}$,
we define the multivariate CUSUM statistic as

\begin{equation*}
S_T(k)
=
\sqrt{\frac{k(T-k)}{T}}
\left(
\frac1k\sum_{t=1}^{k}Z_t
-
\frac1{T-k}\sum_{t=k+1}^{T}Z_t
\right).
\end{equation*}

The scaling factor
$\sqrt{k(T-k)/T}$
is used to balance the fluctuation of the statistic across different candidate locations, thereby reducing the variance differences caused by change points located near the sequence boundaries.

Furthermore, we define the global test statistic and the corresponding change-point estimator as

\begin{equation*}
W_T
=
\max_{1\leq k<T}
\|S_T(k)\|,
\qquad
\widehat{k}
\in
\arg\max_{1\leq k<T}
\|S_T(k)\|.
\end{equation*}

Under the null hypothesis $H_0$, the embedding means corresponding to all network snapshots remain stable, and therefore $W_T$ mainly reflects variations caused by the DeepWalk representation learning procedure and random network fluctuations.
Under the alternative hypothesis $H_1$, the embedding means before and after the change point are different, causing $S_T(k)$ to exhibit a significant peak around the true change point.
Therefore, maximizing the CUSUM statistic enables the localization of structural change points in dynamic networks.
Figure~\ref{fig1} summarizes the overall workflow of the proposed method.
The corresponding algorithm is presented in Algorithm~\ref{alg:deepwalk-cpd}.

\begin{figure}[t]
\centering
\includegraphics[width=\columnwidth]{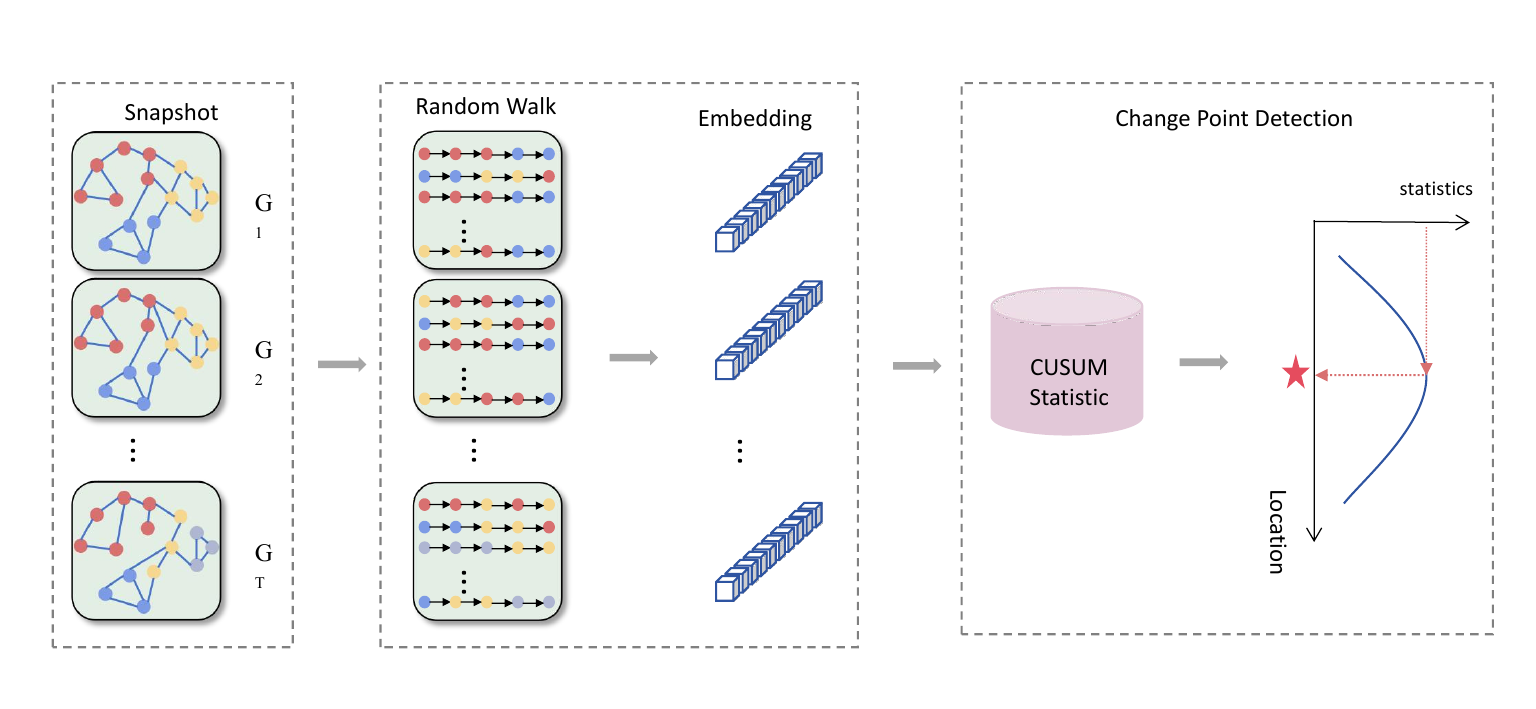}
\caption{The workflow of the proposed DeepWalk-based dynamic network change-point detection method, including fixed-reference coordinate alignment before graph-level pooling.}
\label{fig1}
\end{figure}

\begin{algorithm}[t]
\caption{DeepWalk--CUSUM}
\label{alg:deepwalk-cpd}
\KwIn{Dynamic networks $\mathcal G_{1:T}=\{G_t=(V_t,E_t)\}_{t=1}^{T}$; walks per node $r$; walk length $L$; context window $T_w$; embedding dimension $d$; optional reference $X_{\rm ref}$.}
\KwOut{Global statistic $W_T$ and estimated change point $\widehat k$.}
\tcp{Snapshot-level representation learning}
\For{$t\gets1$ \KwTo $T$}{
  Construct $P_t=D_{A_t}^{-1}A_t$ from $G_t$\;
  Generate $r$ random walks of length $L$ from every node\;
  Construct the co-occurrence matrix $C_t$ using window $T_w$\;
  Learn $X_t\in\mathbb R^{|V_t|\times d}$ from the DeepWalk objective\;
}
\If{$X_{\rm ref}$ is not supplied}{
  $X_{\rm ref}\gets X_1$\;
}
\tcp{Fixed-reference alignment and graph-level pooling}
\For{$t\gets1$ \KwTo $T$}{
  Set anchors $\mathcal A_t\gets V_t\cap V_{\rm ref}$\;
  Compute $X_{t,\mathcal A_t}^{\top}X_{{\rm ref},\mathcal A_t}
  =U_t\Sigma_tV_t^{\top}$\;
  $\widehat Q_t\gets U_tV_t^{\top}$ and $\widetilde X_t\gets X_t\widehat Q_t$\;
  $Z_t\gets |V_t|^{-1}\sum_{i\in V_t}\widetilde X_{t,i}$\;
}
\tcp{CUSUM scan}
\For{$k\gets1$ \KwTo $T-1$}{
  $S_T(k)\gets\sqrt{k(T-k)/T}\,
  \bigl(\overline Z_{1:k}-\overline Z_{k+1:T}\bigr)$\;
}
$W_T\gets\max_{1\le k<T}\|S_T(k)\|$\;
$\widehat k\gets\arg\max_{1\le k<T}\|S_T(k)\|$\;
\Return $(W_T,\widehat k)$\;
\end{algorithm}

\section{Theoretical Analysis}
\label{sec:theory}

This section establishes a theoretical framework for dynamic network change-point detection based on DeepWalk representations, with the goal of characterizing how random perturbations arising from the representation learning process affect the final CUSUM scan statistic.
Since DeepWalk embeddings are generated jointly by random-walk sampling and nonlinear optimization, their statistical properties cannot be directly obtained from the original network structures.
Therefore, we develop the theoretical analysis from two perspectives: the population representation structure and finite-sample perturbations.

Specifically, we first construct the population DeepWalk co-occurrence operator and characterize the dominant geometric properties of the population embedding representation through its spectral structure.
Then, we analyze the empirical operator perturbation caused by finite network size and finite random-walk sampling, and further establish the propagation relationship between such perturbations and the stochastic fluctuations of graph-level embeddings.
Based on this analysis, we study the impact of representation fluctuations on the CUSUM statistic and prove that reliable change-point detection and consistent localization can be achieved when the deterministic signal induced by structural changes in the embedding space dominates the random fluctuation scale.

Therefore, this paper establishes a complete error propagation mechanism from network structures, random-walk operators, DeepWalk representation perturbations, to CUSUM statistics, providing theoretical foundations for understanding the relationship between representation stability and dynamic network change-point detection performance.
We first present all basic assumptions required for the theoretical analysis.

\begin{assumption}
\label{ass:overall}

The following conditions hold.

\begin{enumerate}

\item[(A1)] \textbf{Spectral gap condition.}
Let
$\widetilde P=D_A^{-1/2}AD_A^{-1/2}=U\Lambda U^\top$
be the spectral decomposition of the symmetrically normalized transition matrix, where
$\Lambda=\mathrm{diag}(\lambda_1,\ldots,\lambda_n)$.
Assume that there exists an integer $K$ such that
$\lambda_K>\lambda_{K+1}$,
and define the spectral gap as
$\gamma_K=\lambda_K-\lambda_{K+1}>0$.

\item[(A2)] \textbf{Small initialization condition.}
Let
$W^{(s)}=(X^{(s)},Y^{(s)})^\top$
denote the node representations and context representations at the $s$-th gradient update of DeepWalk.
Assume that the initial embedding satisfies
$\|W^{(0)}\|_F=o(1)$.

\item[(A3)] \textbf{Embedding stochastic fluctuation condition.}
Under the null hypothesis, the graph-level embedding representation satisfies
$Z_t=\mu+\varepsilon_t$,
where
$\mu=\mathbb E(Z_t)$,
and the random error sequence
$\{\varepsilon_t\}_{t=1}^{T}$
is independent with
$\mathbb E(\varepsilon_t)=0$.
Furthermore, there exists a constant
$\sigma_\varepsilon>0$
such that for any
$t=1,\ldots,T$,
any unit vector
$u\in\mathbb S^{d-1}$,
and any
$\lambda\in\mathbb R$,
$\mathbb E\exp
\left(
\lambda
\left\langle u,\varepsilon_t\right\rangle
\right)
\leq
\exp
\left(
\frac{\lambda^2\sigma_\varepsilon^2}{2}
\right).$
In other words, the projection
$\langle u,\varepsilon_t\rangle$
of $\varepsilon_t$ onto any unit direction is a sub-Gaussian random variable, and its sub-Gaussian parameter is uniformly controlled by the same constant
$\sigma_\varepsilon$.

\item[(A4)] \textbf{Interior change-point condition.}
There exists an unknown change point
$k^*\in\{1,\ldots,T-1\}$,
such that the embedding mean satisfies
$\mu_t=\mu^{(1)}$ for $t\leq k^*$,
and
$\mu_t=\mu^{(2)}$ for $t>k^*$.
Moreover, the change point satisfies an interior condition: there exists a constant
$\tau_0\in(0,1/2)$
such that
$\tau_0T\leq k^*\leq(1-\tau_0)T,$
and define the signal strength as
$\Delta=\|\mu^{(1)}-\mu^{(2)}\|>0 .$
\end{enumerate}

\end{assumption}

Assumption~\ref{ass:overall} summarizes the key conditions required for the theoretical analysis of DeepWalk representation learning and change-point detection.
Condition (A1) controls the spectral structure of the random-walk transition matrix.
Since DeepWalk representations depend on high-order neighborhood information accumulated during random walks, the spectral gap condition ensures that dominant structural information can be stably separated from other spectral components.

Condition (A2) characterizes the local stability of the DeepWalk optimization process.
Under the small initialization condition, the nonlinear gradient updating process can be analyzed through the population co-occurrence operator, thereby establishing the theoretical connection between network structures and embedding representations.

Condition (A3) describes the source of stochastic fluctuations in graph-level embedding representations.
This condition jointly characterizes finite network size, random-walk sampling, and representation learning errors as stochastic perturbations in the embedding space, providing the foundation for subsequent concentration analysis of the CUSUM statistic.

Condition (A4) is a standard structural condition in change-point detection.
The interior change-point assumption prevents the true change point from being too close to the boundaries of the observation interval, ensuring that the CUSUM statistic forms a stable and identifiable peak around the true change point.

Overall, these assumptions respectively control the spectral stability, optimization stability, random errors, and detectability of structural changes in DeepWalk representations, providing theoretical foundations for establishing representation perturbation bounds, CUSUM fluctuation bounds, and consistency of change-point localization.

\subsection{Population DeepWalk Operator and Spectral Representation}

This section studies the population structure underlying the DeepWalk embedding process.
The main objective is to characterize the implicit network operator learned by DeepWalk and establish its relationship with the spectral geometry of the underlying network.

Consider a static graph
$G=(V,E)$,
with adjacency matrix
$A\in\mathbb R^{n\times n}$,
and the corresponding degree matrix
$D_A=\mathrm{diag}(A\mathbf1).$
The random-walk transition matrix is defined as
$P=D_A^{-1}A.$
DeepWalk generates random-walk sequences on the graph according to the transition matrix $P$.
Specifically, when a random walk is located at node $i$, it moves to node $j$ with probability $P_{ij}$.
Therefore, after $t$ steps of random walk, the transition probability from node $i$ to node $j$ is given by
$(P^t)_{ij}$.
To formally describe this process, we introduce the following population co-occurrence operator.

\begin{definition}[Population DeepWalk Operator]
\label{def:population_deepwalk}

Let $T_w$ be the context window size of DeepWalk and $L$ be the random-walk length.
The population DeepWalk operator is defined as

\[
M=
2\sum_{t=1}^{T_w}(L-t)\Pi P^t,
\]
where
$\Pi=D_A/\mathrm{vol}(G)$
is the stationary distribution matrix of the random walk,
and
$\mathrm{vol}(G)=\sum_{i=1}^{n}d_i .$

\end{definition}

The population DeepWalk operator $M$ characterizes the expected node co-occurrence structure generated by the random-walk mechanism.
More specifically, the $(i,j)$-th element of matrix $M$ represents the expected number of times that nodes $i$ and $j$ co-occur within the given window range along random-walk trajectories.
The optimization objective of DeepWalk essentially performs an implicit factorization of this population co-occurrence operator.

The population DeepWalk operator integrates multiple powers of the random-walk transition matrix
$P,P^2,\ldots,P^{T_w}$.
Therefore, DeepWalk can capture not only local adjacency relationships but also high-order connectivity patterns and latent community structures in networks.
We next establish the relationship between the DeepWalk operator and the spectral structure of graph random walks.

Since the random-walk matrix $P$ is generally not symmetric, we define the symmetrically normalized transition operator as

\[
\widetilde P
=
D_A^{1/2}PD_A^{-1/2}
=
D_A^{-1/2}AD_A^{-1/2}.
\]

The matrix $\widetilde P$ is symmetric and therefore admits the spectral decomposition
$\widetilde P
=
U\Lambda U^\top,$
where
$\Lambda=\mathrm{diag}(\lambda_1,\ldots,\lambda_n)$
is the eigenvalue matrix satisfying
$1=\lambda_1\geq \lambda_2\geq\cdots\geq\lambda_n,$
and
$U=(u_1,\ldots,u_n)$
is the corresponding orthogonal eigenvector matrix.

Using the above spectral decomposition, the population DeepWalk operator can be further represented as

\[
M=
2\sum_{t=1}^{T_w}(L-t)
D_A^{1/2}
U\Lambda^tU^\top
D_A^{-1/2}\Pi .
\]

It can be seen that the spectral structure of the DeepWalk operator is determined by the dominant eigenvectors of the random-walk transition matrix.
In particular, when the network contains latent community structures, the leading eigenvectors of the random-walk transition matrix can characterize the large-scale connectivity patterns of the network.
Therefore, DeepWalk embeddings can approximately recover the dominant spectral subspace determined by the underlying network structure.

Let
$U_K=\mathrm{span}(u_1,\ldots,u_K)$
denote the leading $K$-dimensional spectral subspace of the random-walk transition operator.
The following proposition establishes the dominant spectral structure of the population DeepWalk operator.

\begin{proposition}[Spectral Structure of the DeepWalk Operator]
\label{prop:deepwalk_spectral}

Suppose condition \textnormal{(A1)} in Assumption~\ref{ass:overall} holds.
Then the leading eigenspace of the population DeepWalk operator $M$ coincides with the leading eigenspace of the graph random-walk transition operator $P$.
More specifically,
$\mathrm{span}(u_1,\ldots,u_K)$
is invariant under the population DeepWalk operator $M$.

\end{proposition}

This proposition shows that DeepWalk does not learn arbitrary low-dimensional representations, but theoretically recovers the dominant spectral geometric structure induced by the random-walk transition operator.
Therefore, when the dynamic network generation mechanism changes, variations in the underlying random-walk operator will further lead to changes in DeepWalk spectral embedding representations, providing the foundation for the subsequent change-point detection theory based on DeepWalk representations.

\subsection{DeepWalk Optimization Dynamics and Embedding Concentration}

The previous section analyzed the spectral structure of the population DeepWalk operator and showed that the main structural information contained in DeepWalk representations originates from the dominant spectral space of the random-walk transition operator.
However, in practical applications, DeepWalk embeddings are generated from finite random-walk samples and obtained through a nonlinear optimization process.
Therefore, it is necessary to further analyze the deviation between empirical DeepWalk representations and their population counterparts.
Let $C$ denote the empirical co-occurrence matrix obtained from finite random-walk sequences, and let $M$ denote the population DeepWalk operator defined in the previous section.
Then the empirical co-occurrence matrix can be expressed as
$C=M+E,$
where $E$ represents the random perturbation caused by finite random-walk sampling.

On the other hand, DeepWalk obtains low-dimensional embeddings by optimizing the objective function between node representations and context representations.
Let $X$ and $Y$ denote the node embedding matrix and context embedding matrix, respectively.
The DeepWalk objective function can be written as
\[
\mathcal L(X,Y)
=
-\sum_{i,j}C_{ij}
\log
\frac{\exp(x_i^\top y_j)}
{\sum_l\exp(x_i^\top y_l)} .
\]
where $x_i$ and $y_j$ represent the embedding vector of node $i$ and the context vector of node $j$, respectively.
Due to the nonlinear nature of the above optimization procedure, directly analyzing its statistical properties is challenging.
Therefore, we utilize the small initialization condition to investigate the local behavior of the DeepWalk gradient updating process around the population operator.
Let $W^{(s)}$ denote the parameter matrix after the $s$-th gradient update, namely,
$W^{(s)}=(X^{(s)},Y^{(s)})^\top .$
Under the small initialization condition, the optimization dynamics of DeepWalk can be expanded locally around the population co-occurrence operator.
This condition guarantees that the nonlinear optimization process of DeepWalk can be theoretically analyzed through local linearization. The following theorem formalizes this local linearization and characterizes the resulting approximation error.

\begin{theorem}[Linearization of DeepWalk Dynamics]
\label{thm:linearized_deepwalk}

Under conditions \textnormal{(A1-A2)} in Assumption~\ref{ass:overall},
the gradient updating process of DeepWalk can be represented as

\[
W^{(s+1)}
=
\mathcal{L}_{M}W^{(s)}
+
R^{(s)},
\]
where
$\mathcal{L}_{M}$
denotes the linear updating operator determined by the population DeepWalk operator $M$,
and the remainder term satisfies

\[
\|R^{(s)}\|_F
=
O(\|W^{(s)}\|_F^2).
\]

\end{theorem}
Therefore, when the initialization is sufficiently small, the nonlinear optimization process of DeepWalk can be approximated by a linear dynamical system driven by the population co-occurrence operator.
This result establishes the connection between the DeepWalk optimization process and the population spectral structure.
Specifically, the training dynamics of DeepWalk are mainly determined by the spectral properties of the population DeepWalk operator in the local region, and therefore the final learned embeddings are able to reflect the underlying network structure.

Furthermore, since practical computation uses the empirical co-occurrence matrix $C$ rather than the population operator $M$, it is necessary to characterize the embedding deviation caused by matrix perturbations.

Let $\widehat U_K$ denote the leading $K$-dimensional spectral subspace of the empirical DeepWalk operator, and let $U_K$ denote the leading spectral subspace of the population DeepWalk operator.
According to spectral perturbation theory, the deviation between the empirical spectral subspace and the population spectral subspace is determined jointly by the magnitude of empirical operator perturbation and the spectral gap.
The following theorem formalizes this perturbation relationship and establishes the concentration rate of the empirical DeepWalk embedding space around its population counterpart.

\begin{theorem}[DeepWalk Embedding Concentration]
\label{thm:embedding_concentration}

Under conditions \textnormal{(A1-A2)} in Assumption~\ref{ass:overall}, if the empirical co-occurrence matrix satisfies the corresponding concentration condition, then the empirical DeepWalk embedding space satisfies

\[
\left\|
\sin
\Theta
(
\widehat U_K,U_K
)
\right\|
=
O_P
\left(
\frac{\|E\|}{\gamma_K}
\right).
\]

In particular, when the empirical co-occurrence matrix perturbation satisfies
$\|E\|
=
O_P
\left(
\sqrt{\frac{\log n}{n\rho}}
\right),$
we have

\[
\left\|
\sin
\Theta
(
\widehat U_K,U_K
)
\right\|
=
O_P
\left(
\frac{1}{\gamma_K}
\sqrt{\frac{\log n}{n\rho}}
\right).
\]

\end{theorem}

This theorem shows that the error of DeepWalk embeddings is mainly determined by the random perturbation of the empirical co-occurrence operator and the sensitivity of the population spectral structure to such perturbations.
Specifically,
$\sqrt{\frac{\log n}{n\rho}}$
characterizes the empirical operator perturbation caused by finite network size and random-walk sampling, while
$\frac{1}{\gamma_K}$
characterizes the amplification effect when such perturbations propagate to the leading embedding subspace.

More specifically, the spectral gap
$\gamma_K=\lambda_K-\lambda_{K+1}$
measures the separation between the $K$-th dominant spectral direction and the remaining spectral components.
When $\gamma_K$ is large, the dominant spectral subspace is clearly separated from the non-dominant spectral subspace.
Therefore, even when the empirical co-occurrence operator is subject to certain random perturbations, its leading $K$-dimensional spectral space remains relatively stable.
In contrast, when $\gamma_K$ is small, two adjacent spectral subspaces become difficult to distinguish, and even small operator perturbations may cause substantial deviations of the eigenspace.

Therefore, a larger spectral gap implies more stable DeepWalk representations, whereas a smaller spectral gap amplifies the errors caused by random-walk sampling.
As the effective network size $n\rho$ increases, random-walk sampling becomes more sufficient, or the spectral gap becomes larger, empirical DeepWalk representations can more accurately recover the population network structure.

This result provides the theoretical foundation for analyzing the stochastic fluctuations of the DeepWalk-CUSUM statistic in the next section.

\subsection{Fluctuation Bound of the DeepWalk-CUSUM Process under the Null Hypothesis}

The previous section established the concentration property of DeepWalk embedding representations, showing that empirical embeddings can recover the population network structure within a certain error range.
This section further investigates how stochastic perturbations in DeepWalk representations affect the CUSUM scan statistic when no structural change exists.

Consider the null hypothesis
$H_0:\mu_1=\mu_2=\cdots=\mu_T .$
Under this hypothesis, graph-level embeddings at different time points share the same population mean.
Therefore, after appropriate rotational alignment, the graph-level representations can be written as
$Z_t=\mu+\varepsilon_t,$
where $\mu$ denotes the population embedding mean during the stationary period, and $\varepsilon_t$ represents stochastic perturbations caused by finite network size, random-walk sampling, and the DeepWalk optimization procedure.

The DeepWalk-CUSUM statistic is defined as

\[
S_T(k)
=
\sqrt{\frac{k(T-k)}{T}}
\left(
\frac{1}{k}\sum_{t=1}^{k}Z_t
-
\frac{1}{T-k}\sum_{t=k+1}^{T}Z_t
\right),
\]
where $k$ denotes a candidate change-point location.
Under the null hypothesis, since the two time intervals have the same mean component,
$\frac{1}{k}\sum_{t=1}^{k}\mu
-
\frac{1}{T-k}\sum_{t=k+1}^{T}\mu
=0.$
Therefore, the deterministic structural change component vanishes in the CUSUM statistic, and the statistic is solely determined by random embedding errors:

\[
S_T(k)
=
\sqrt{\frac{k(T-k)}{T}}
\left(
\frac{1}{k}\sum_{t=1}^{k}\varepsilon_t
-
\frac{1}{T-k}\sum_{t=k+1}^{T}\varepsilon_t
\right).
\]

It can be seen that, in the absence of change points, the fluctuation of DeepWalk-CUSUM is entirely caused by stochastic errors in the representation learning process.
Define the maximum scan statistic as
$W_T
=
\max_{\Delta_T\leq k\leq T-\Delta_T}
\|S_T(k)\|,$
where $\Delta_T$ is introduced to prevent candidate change points from being too close to the sequence boundaries.

Since the embedding errors satisfy condition \textnormal{(A3)} in Assumption~\ref{ass:overall}, namely the uniform directional sub-Gaussian condition, vector concentration inequalities and maximal deviation bounds over candidate locations can be applied to control the entire CUSUM scanning process.
It should be noted that condition \textnormal{(A3)} is not directly implied by the previous spectral subspace concentration results.
Instead, it provides an additional probabilistic characterization of the tail behavior of graph-level embedding errors.
Specifically, $\varepsilon_t$ jointly incorporates random errors arising from finite network size, random-walk sampling, empirical co-occurrence matrix perturbations, DeepWalk optimization, and node representation aggregation.
Under appropriate regularity conditions, these errors generated by the accumulation of many local random fluctuations typically exhibit light-tailed behavior.
Therefore, the uniform directional sub-Gaussian condition provides a natural characterization.

This condition allows simultaneous control of fluctuations in arbitrary projection directions and further leads to a uniform upper bound for the multivariate CUSUM statistic.The following theorem formalizes this uniform control and establishes the stochastic order of the DeepWalk-CUSUM statistic under the null hypothesis.

\begin{theorem}[Fluctuation Bound of DeepWalk-CUSUM under the Null Hypothesis]
\label{thm:null_fluctuation}

Suppose conditions \textnormal{(A1-A3)} in Assumption~\ref{ass:overall} hold.
Then, under the null hypothesis $H_0$,

\[
W_T
=
O_P\left(
\frac{1}{\gamma_K}
\sqrt{\frac{\log n}{n\rho}}
\sqrt{d+\log T}
\right).
\]

In particular, when the embedding dimension $d$ is fixed,

\[
W_T
=
O_P\left(
\frac{1}{\gamma_K}
\sqrt{\frac{\log n\log T}{n\rho}}
\right).
\]

\end{theorem}

The theorem characterizes the maximum stochastic fluctuation range of the DeepWalk-CUSUM statistic over all candidate change-point locations under the absence of structural changes.
The different factors in the obtained upper bound have clear statistical interpretations.

First, the term
$\sqrt{\frac{\log n}{n\rho}}$
originates from the concentration error of the empirical DeepWalk co-occurrence operator relative to the population co-occurrence operator.
Here, $n\rho$ can be interpreted as the effective connectivity scale of the network.
A larger number of nodes or a denser network provides more sufficient structural information for random walks, making empirical co-occurrence relationships closer to their population expectations.
The factor $\log n$ arises from the logarithmic complexity required to uniformly control high-dimensional network operators.

Second, the factor
$\frac{1}{\gamma_K}$
comes from spectral subspace perturbation.
The random error of the empirical co-occurrence operator propagates into the DeepWalk embedding space through spectral decomposition, and a smaller spectral gap leads to a stronger amplification effect of this perturbation.

Finally,
$\sqrt{d+\log T}$
reflects the complexity introduced by both the embedding dimension and temporal scanning.
Specifically, $\sqrt d$ represents the dimensional cost required to control the norm of a $d$-dimensional random vector, whereas $\sqrt{\log T}$ arises from taking the maximum over all candidate locations.
Since DeepWalk-CUSUM compares approximately $T-1$ candidate change-point locations, uniformly controlling these random statistics introduces an additional logarithmic scanning cost.
Therefore, when $d$ is fixed, the temporal scanning component mainly appears as the factor $\sqrt{\log T}$.

Different from classical CUSUM methods, the stochastic errors considered in this paper do not directly originate from the original network observations.
Instead, they are jointly induced by random-walk sampling, DeepWalk representation learning, and graph-level aggregation.
Therefore, this result establishes a theoretical connection among network structures, representation stability, and downstream change-point detection statistics.

Furthermore, in the presence of structural changes, the CUSUM statistic can reliably distinguish true structural transitions from representation noise only when the change signal strength exceeds the above stochastic fluctuation scale.
The next section establishes the detection boundary and localization consistency results for DeepWalk-CUSUM based on this observation.

\subsection{Detection Boundary and Localization Consistency}

The previous section analyzed the stochastic fluctuation range of the DeepWalk-CUSUM statistic under the null hypothesis.
This section further studies the behavior of the statistic under structural changes and establishes the consistency of change-point detection together with the localization error bound.

Under the alternative hypothesis, suppose that there exists a unique change point $k^*$ such that the graph-level embedding mean satisfies the piecewise structure

\[
\mu_t=
\begin{cases}
\mu^{(1)}, & t\leq k^*,\\
\mu^{(2)}, & t>k^* .
\end{cases}
\]

Define the change magnitude in the embedding space as
$\Delta=
\|\mu^{(1)}-\mu^{(2)}\|.$
Then the graph-level embedding representation can be written as
$Z_t=\mu_t+\varepsilon_t,$
where $\mu_t$ represents the deterministic signal induced by network structural changes, and $\varepsilon_t$ represents stochastic perturbations generated during the DeepWalk representation process.

The CUSUM statistic can be decomposed into the mean change component and the random error component:
$S_T(k)
=
S_T^{(\mu)}(k)
+
S_T^{(\varepsilon)}(k),$
where

\[
S_T^{(\mu)}(k)
=
\sqrt{\frac{k(T-k)}{T}}
\left(
\frac{1}{k}\sum_{t=1}^{k}\mu_t
-
\frac{1}{T-k}\sum_{t=k+1}^{T}\mu_t
\right),
\]
represents the deterministic signal generated by structural changes,
and

\[
S_T^{(\varepsilon)}(k)
=
\sqrt{\frac{k(T-k)}{T}}
\left(
\frac{1}{k}\sum_{t=1}^{k}\varepsilon_t
-
\frac{1}{T-k}\sum_{t=k+1}^{T}\varepsilon_t
\right),
\]
represents the stochastic fluctuation caused by DeepWalk representation errors.

We first analyze the properties of the deterministic signal component.The following lemma characterizes the shape of the deterministic CUSUM signal and quantifies its decay away from the true change point.

\begin{lemma}[Deterministic Signal Structure of the DeepWalk-CUSUM Process]
\label{lem:signal_structure}

Assume condition \textnormal{(A4)} in Assumption~\ref{ass:overall} holds.
Then
$\|S_T^{(\mu)}(k)\|$
achieves its unique maximum at
$k=k^*$, and
$\|S_T^{(\mu)}(k^*)\|
=
\sqrt{
\frac{k^*(T-k^*)}{T}
}
\Delta.$
Furthermore, there exists a constant
$c_0>0$
depending only on $\tau_0$, such that for all
$1\leq k<T$,

\[
\|S_T^{(\mu)}(k^*)\|
-
\|S_T^{(\mu)}(k)\|
\geq
c_0
\Delta
\frac{|k-k^*|}{\sqrt{T}}.
\]

\end{lemma}

This lemma shows that structural changes generate a peak in the CUSUM curve centered at the true change point, and the deterministic signal becomes weaker as the candidate location moves farther away from the true change point.

On the other hand, according to the result in Section 3.3, the random error component satisfies
$\max_k
\|S_T^{(\varepsilon)}(k)\|
=
O_P(R_T).$
Therefore, change-point detection depends on comparing the deterministic signal strength with the stochastic fluctuation scale.

Define the DeepWalk-CUSUM change-point estimator as
$\widehat{k}
=
\arg\max_k
\|S_T(k)\|.$
The following theorem establishes the detection boundary.

\begin{theorem}[Detection Boundary of DeepWalk-CUSUM]
\label{thm:detection_boundary}

Suppose conditions \textnormal{(A1)}--\textnormal{(A4)}
in Assumption~\ref{ass:overall} hold.
Then the statistic
$W_T
=
\max_{1\leq k<T}
\|S_T(k)\|$
satisfies

\[
W_T
\geq
\sqrt{
\frac{k^*(T-k^*)}{T}
}
\Delta
-
O_P\left(
\frac{1}{\gamma_K}
\sqrt{
\frac{\log n}{n\rho}
}
\sqrt{
d+\log T
}
\right).
\]

Therefore, if
$\sqrt{
\frac{k^*(T-k^*)}{T}
}
\Delta
\gg
\frac{1}{\gamma_K}
\sqrt{
\frac{\log n}{n\rho}
}
\sqrt{
d+\log T
},$
then
$W_T
\overset{P}{\longrightarrow}
\infty.$
\end{theorem}

This theorem shows that the detection capability of DeepWalk-CUSUM depends on the relative magnitude between the deterministic change signal and the stochastic fluctuation of representations.
The deterministic peak at the true change point has magnitude
$\sqrt{\frac{k^*(T-k^*)}{T}}\Delta,$
whereas the maximum stochastic fluctuation under the null hypothesis has magnitude
$\frac{1}{\gamma_K}
\sqrt{\frac{\log n}{n\rho}}
\sqrt{d+\log T}.$

Therefore, reliable separation between true structural changes and DeepWalk representation noise can be achieved only when the former asymptotically dominates the latter.

The detection boundary also reveals which types of dynamic networks are easier to analyze.
First, a larger change magnitude $\Delta$ creates a more significant mean shift in the embedding space and is therefore easier to detect.
Second, a larger spectral gap $\gamma_K$ indicates that the dominant spectral structure of the network is more stable, reducing the propagation of empirical operator perturbations into the embedding space.
Third, a larger effective network size $n\rho$ indicates that the network contains more sufficient connectivity information, thereby reducing the stochastic error of empirical DeepWalk representations.

Conversely, when the network is overly sparse, the dominant spectral space is not well separated from other spectral components, or the selected embedding dimension is large, the stochastic fluctuations of DeepWalk representations increase, requiring stronger structural changes for reliable detection.
In addition, although increasing the temporal length $T$ provides more network snapshots, it also increases the number of candidate change-point locations and therefore introduces an additional logarithmic scanning cost.

Furthermore, we can obtain the localization error bound for the estimated change point.

\begin{theorem}[Consistency of Change-point Localization]
\label{thm:localization}

Suppose conditions \textnormal{(A1)}--\textnormal{(A4)}
in Assumption~\ref{ass:overall} hold.
If

\[
\sqrt{T}\Delta
\gg
\frac{1}{\gamma_K}
\sqrt{
\frac{\log n}{n\rho}
}
\sqrt{
d+\log T
},
\]

then

\[
\frac{
|\widehat{k}-k^*|
}{T}
\overset{P}{\longrightarrow}
0.
\]

More precisely,

\[
|\widehat{k}-k^*|
=
O_P\left(
\frac{
\sqrt{T}
}{
\gamma_K\Delta
}
\sqrt{
\frac{\log n}{n\rho}
}
\sqrt{
d+\log T
}
\right).
\]

\end{theorem}

This theorem provides a stronger guarantee than detection consistency.
The detection boundary only indicates that the statistic can exceed the stochastic fluctuation level under the alternative hypothesis, thereby determining whether a structural change exists in the network sequence.
In contrast, localization consistency further requires that the relative error between the estimated change point $\widehat{k}$ and the true change point $k^*$ converges to zero, namely,
$\frac{|\widehat{k}-k^*|}{T}
\overset{P}{\longrightarrow}
0.$
The source of the precise localization error bound can be explained by the local shape of the deterministic CUSUM curve.
According to the deterministic signal lemma above, when the candidate location $k$ deviates from the true change point $k^*$, the deterministic CUSUM signal decreases by at least
$c_0\Delta
\frac{|k-k^*|}{\sqrt T}.$

On the other hand, the maximum perturbation of the random CUSUM process over all candidate locations is
$O_P(R_T),$
where
$R_T=
\frac{1}{\gamma_K}
\sqrt{\frac{\log n}{n\rho}}
\sqrt{d+\log T}.$
To ensure that the statistic at the true change point remains larger than that at deviating locations, the decrease of the deterministic signal must be comparable to the stochastic perturbation.
Therefore, by setting
$\Delta
\frac{|\widehat{k}-k^*|}{\sqrt T}
=
O_P(R_T),$
we obtain
$|\widehat{k}-k^*|
=
O_P
\left(
\frac{\sqrt T R_T}{\Delta}
\right),$
which gives the localization error bound in the theorem.

This result further indicates that localization accuracy improves with increasing change magnitude $\Delta$, spectral gap $\gamma_K$, and effective network size $n\rho$.
Higher embedding dimensions $d$ and a larger number of candidate time points increase stochastic fluctuations and therefore reduce localization accuracy.
It should be noted that the absolute localization error does not necessarily converge to zero as $T$ increases.
However, under the given signal condition, its ratio relative to the sequence length converges to zero.
Therefore, the obtained result corresponds to relative localization consistency.

In summary, this section establishes the theoretical connection between stochastic errors in DeepWalk representations and change-point detection performance.
Specifically, structural changes in the network first induce mean changes in the embedding space, which generate deterministic signals in the CUSUM statistic.
Meanwhile, finite network size and random-walk sampling introduce stochastic fluctuations in the representations.
When the structural change signal exceeds the stochastic fluctuation scale, DeepWalk-CUSUM achieves reliable detection and further obtains consistent change-point localization.

\section{Extension to Multiple Change Points}
\label{sec:multiple}

Real-world dynamic networks often contain multiple structural transitions.
Suppose the true change points satisfy
$1<k_1^*<\cdots<k_m^*<T$,
and let
$k_0^*=0$
and
$k_{m+1}^*=T$.
We assume that the embedding mean is piecewise constant, that is,
$\mu_t=\mu^{(j)}$
when
$k_{j-1}^*<t\leq k_j^*$,
and adjacent segments satisfy
$\mu^{(j)}\neq\mu^{(j+1)}$.

\subsection{Recursive Binary Segmentation}

For any candidate interval $[s,e]$, define the local CUSUM statistic as

\begin{align*}
S_{s,e}(k)={}&
\sqrt{\frac{(k-s+1)(e-k)}{e-s+1}}
\notag\\
&\times
\left(
\frac1{k-s+1}\sum_{t=s}^k Z_t
-
\frac1{e-k}\sum_{t=k+1}^e Z_t
\right).
\end{align*}

Here, $s\leq k<e$, and define
$W_{s,e}
=
\max_{s\leq k<e}
\|S_{s,e}(k)\|.$
If $W_{s,e}$ exceeds the interval-specific threshold
$\zeta_{s,e}$,
the interval is split at

\begin{equation*}
\widehat k
=
\arg\max_{s\leq k<e}
\|S_{s,e}(k)\|,
\end{equation*}
and the subintervals
$[s,\widehat k]$
and
$[\widehat k+1,e]$
are processed recursively.
Otherwise, the segmentation procedure stops.
The final estimated change-point set is
$\widehat{\mathcal K}
=
\{\widehat k_1,\ldots,\widehat k_{\widehat m}\}.$

This procedure recursively reduces the multiple change-point problem to a sequence of single change-point problems.
Therefore, the single change-point theory developed in Section~\ref{sec:theory} provides the theoretical foundation for this procedure.

\subsection{Threshold Selection and Implementation Details}

According to Theorem~\ref{thm:null_fluctuation}, for an interval $[s,e]$, the threshold can be chosen as

\begin{equation*}
\zeta_{s,e}
=
C
\frac1{\gamma_K}
\sqrt{\frac{\log n}{n\rho}}
\sqrt{d+\log(e-s+1)},
\end{equation*}
where $C>0$ is a tuning constant.
To avoid unstable estimation on excessively short intervals, segmentation is continued only when
$e-s+1\geq h$,
where $h$ is the minimum interval length.
A smaller value of $h$ increases sensitivity to short-lived changes, whereas a larger value of $h$ improves robustness against noise.

\subsection{Computational Complexity}

If each node generates $r$ random walks of length $L$, the computational cost of random-walk generation for all $T$ snapshots is approximately
$O(TnrL).$
After obtaining the graph-level embeddings, all CUSUM candidate values on an interval of length $m$ can be computed in
$O(md)$
time using cumulative sums.
The total complexity of the binary segmentation stage is approximately
$O(Td\log T).$
Therefore, the overall computational cost is mainly determined by the DeepWalk embedding stage.
The embeddings of different network snapshots can be computed independently, making the proposed method suitable for parallel implementation.

\begin{algorithm}[t]
\caption{Recursive Binary Segmentation}
\label{alg:binary-segmentation}
\KwIn{Graph-level sequence $\{Z_t\}_{t=1}^{T}$; initial interval $[s,e]$; interval threshold $\zeta_{a,b}$; minimum interval length $h$.}
\KwOut{Estimated change-point set $\widehat{\mathcal K}$.}
$\widehat{\mathcal K}\gets\emptyset$\;
\Fn{\Segment{$a,b$}}{
  \If{$b-a+1<h$}{
    \Return\;
  }
  Compute $S_{a,b}(k)$ for every $a\le k<b$\;
  $W_{a,b}\gets\max_{a\le k<b}\|S_{a,b}(k)\|$\;
  \If{$W_{a,b}>\zeta_{a,b}$}{
    $\widehat k\gets\arg\max_{a\le k<b}\|S_{a,b}(k)\|$\;
    $\widehat{\mathcal K}\gets\widehat{\mathcal K}\cup\{\widehat k\}$\;
    \Segment{$a,\widehat k$}\;
    \Segment{$\widehat k+1,b$}\;
  }
}
\Segment{$s,e$}\;
Sort $\widehat{\mathcal K}$ in ascending order\;
\Return $\widehat{\mathcal K}$\;
\end{algorithm}

\section{Simulation Studies}
\label{sec:simulation}
This section evaluates the proposed method through a sequence of experiments that separate detection, localization, robustness, and computational behavior. The primary experiments use density-controlled structural changes so that performance cannot be attributed to a change in the expected number of edges.

\subsection{Experimental Setup}

\textit{Data-generating mechanisms.}
Unless otherwise stated, snapshots are independent conditional on the regime, node identities are fixed over time, and a change replaces a three-community stochastic block model (SBM) by a five-community SBM. For arbitrary $n$, community sizes are balanced, meaning that they differ by at most one. If the block sizes are $n_1,\ldots,n_K$, the expected edge density is
\begin{equation}
	\bar p_K=
	\frac{
	\left\{\sum_{g=1}^{K}\binom{n_g}{2}\right\}p_{\rm in}
	+\left\{
	\binom{n}{2}-\sum_{g=1}^{K}\binom{n_g}{2}
	\right\}p_{\rm out}
	}{\binom{n}{2}}.
	\label{eq:density-matched-sbm}
\end{equation}
The post-change within-community probability is recomputed from~\eqref{eq:density-matched-sbm} for each network size, ensuring exact equality of the pre- and post-change expected densities. The main experiment has $n=150$, for which the balanced communities are equal sized.

\textit{Competing methods:}
We compare DeepWalk--CUSUM with adjacency spectral embedding (ASE)--CUSUM, normalized Laplacian spectral embedding (LSE)--CUSUM, Node2Vec--CUSUM based on biased random-walk embeddings~\citep{grover2016node2vec}, the random-walk-in-latent-space descriptor of Lin et al.~\citep{lin2023network}, Laplacian anomaly detection (LAD)~\citep{huang2020laplacian}, scalable change-point detection based on the density of states (SCPD)~\citep{huang2023fast}, and edge-count CUSUM. ASE, LSE, Node2Vec, and RW-LS use graph-level vector summaries followed by the same CUSUM rule as DeepWalk. LAD and SCPD retain their native spectrum-based window scores. Edge-count CUSUM is a density-only diagnostic.

\textit{Implementation details.}
All embedding methods use dimension $d=8$. DeepWalk uses ten walks of length 30 per node and context-window size $T_w=10$. Node2Vec uses the same walk budget with return parameter $p=1$ and in--out parameter $q=0.5$. Every independently fitted node embedding is aligned to the fixed reference $X_{\rm ref}=X_1$ by~\eqref{eq:procrustes} before mean pooling. LAD and SCPD use retrospective reference windows of five snapshots in the calibrated signal, change-location, and runtime experiments, and windows of five and ten snapshots in the main localization and multi-change comparisons. SCPD uses 20 density-of-states bins. Unless otherwise stated, methods are evaluated on paired graph sequences and use method-specific random-number streams derived from fixed replication seeds. The experiments were run in Python 3.12.3 with NumPy 1.26.4 and pandas 2.2.2 on a 16-logical-core Intel processor. Parallel Monte Carlo experiments use up to ten worker processes; the runtime benchmark instead uses one process and one BLAS thread.

For the primary thresholded experiments, a method-specific critical value is the higher empirical 95th percentile of maximum method-specific scan scores from an independent null sample. The calibrated signal-strength experiment uses $B_{\rm cal}=999$ null sequences and $R=200$ evaluation sequences per signal level. The theory-guided diagnostics use $B_{\rm cal}=499$ and $R=100$ per setting. Null calibration sequences are not reused as evaluation sequences.

\textit{Evaluation metrics.}
For an unthresholded localization diagnostic, every method returns the maximizer $\widehat k$ of its scan statistic, and we report
\begin{equation}
\begin{aligned}
{\rm MAE}
&=R^{-1}\sum_{r=1}^R|\widehat k_r-k^*|,\\
{\rm Within5}
&=R^{-1}\sum_{r=1}^R\mathbf 1\{|\widehat k_r-k^*|\le5\}.
\end{aligned}
\label{eq:simulation-metrics}
\end{equation}
These forced-localization measures do not by themselves establish detection. In calibrated experiments, empirical size and power are computed from thresholded decisions. Localization metrics under an alternative are then evaluated conditional on rejection. Multiple-change experiments additionally report the estimated change count and the Hausdorff distance between the estimated and true change-point sets. Runtime includes representation construction and the downstream scan unless stated otherwise.

\subsection{Density-Controlled Single-Change Detection}

The main design has $n=150$, $T=100$, and $k^*=50$. Before the change, each snapshot follows a three-community SBM; after the change, it follows a five-community SBM.
We set $(p_{\rm in}^{(0)},p_{\rm out}^{(0)})=(0.250,0.050)$ before the change and $(p_{\rm in}^{(1)},p_{\rm out}^{(1)})=(0.388,0.050)$ after the change. These parameters give $\bar p_3\approx\bar p_5\approx0.1158$. Consequently, the signal is a reorganization of edges across communities rather than a first-order density shift; an edge-count detector is included only as a density-only diagnostic.

\begin{table}[t]
	\centering
	\caption{Localization results under the density-matched single-change
		setting based on $R=100$ replications. Parentheses report Monte Carlo
		standard errors. Smaller MAE and median absolute error and larger
		Within 5 are preferred; the best value in each column is bold.}
	\label{tab:density-matched-single}
	\begingroup
	\footnotesize
	\setlength{\tabcolsep}{4.0pt}
	\renewcommand{\arraystretch}{1.08}
	\begin{tabular}{@{}lccc@{}}
		\toprule
		Method
		& MAE
		& \shortstack{Median\\AE}
		& \shortstack{Within\\5} \\
		\midrule
		DeepWalk--CUSUM
		& \textbf{1.12 (0.16)} & \textbf{1.0}  & \textbf{0.97 (0.02)} \\
		Node2Vec--CUSUM
		& 4.80 (0.88) & 1.5  & 0.74 (0.04) \\
		LSE--CUSUM
		& 7.85 (1.26) & 2.0  & 0.67 (0.05) \\
		ASE--CUSUM
		& 9.97 (1.40) & 3.5  & 0.59 (0.05) \\
		SCPD
		& 11.84 (1.46) & 3.0  & 0.54 (0.05) \\
		LAD
		& 18.84 (1.25) & 18.0 & 0.20 (0.04) \\
		RW-LS--CUSUM
		& 38.46 (1.16) & 42.5 & 0.03 (0.02) \\
		Edge-count CUSUM
		& 33.84 (1.40) & 37.0 & 0.05 (0.02) \\
		\bottomrule
	\end{tabular}
	\endgroup
\end{table}

Table~\ref{tab:density-matched-single} shows that DeepWalk--CUSUM attains the lowest MAE (1.12) and the largest Within-5 rate (0.97) among the low-dimensional and graph-level methods in this comparison. Node2Vec--CUSUM is the closest random-walk embedding comparator, although its MAE is higher and its Within-5 rate is lower under the same embedding dimension. RW-LS, LAD, and SCPD are less accurate for this density-matched community reallocation. The poor edge-count result is consistent with the intended density-controlled design.

Localization accuracy is not by itself evidence of a level-$\alpha$ test. We therefore applied the primary split-sample calibration protocol at nominal level $\alpha=0.05$. Table~\ref{tab:size-power} reports empirical size under the no-change endpoint and power under the full density-matched alternative. Localization is evaluated only among rejected alternative sequences.

\begin{table*}[t]
\centering
\caption{Calibrated size, power, and conditional localization at nominal level $\alpha=0.05$. Critical values use $B_{\rm cal}=999$ independent null sequences; empirical size and power use $R=200$ additional sequences each. Parentheses report Monte Carlo standard errors. Bold identifies the best result in each column. Critical values are not ranked.}
\label{tab:size-power}
\begin{tabular}{lccccc}
\toprule
Method & Critical value & Size & Power & Alt. MAE & Alt. Within 5 \\
\midrule
DeepWalk--CUSUM & 0.0111 & 0.030 (0.012) & \textbf{1.000 (0.000)} & \textbf{1.12} & \textbf{0.96} \\
Node2Vec--CUSUM & 0.0096 & 0.025 (0.011) & 0.830 (0.027) & 2.99 & 0.89 \\
ASE--CUSUM & 0.0240 & \textbf{0.050 (0.015)} & 0.620 (0.034) & 2.73 & 0.87 \\
LSE--CUSUM & 0.0026 & \textbf{0.050 (0.015)} & 0.750 (0.031) & 2.19 & 0.89 \\
SCPD & 0.0065 & 0.070 (0.018) & 0.150 (0.025) & 6.53 & 0.77 \\
LAD & 0.0008 & 0.025 (0.011) & 0.020 (0.010) & 36.00 & 0.00 \\
RW-LS--CUSUM & 0.0120 & 0.060 (0.017) & 0.025 (0.011) & 45.60 & 0.00 \\
Edge-count CUSUM & 101.5961 & 0.045 (0.015) & 0.030 (0.012) & 36.00 & 0.00 \\
\bottomrule
\end{tabular}
\end{table*}

Table~\ref{tab:size-power} shows that the formal calibration yields empirical sizes between 0.025 and 0.070 across the compared methods. DeepWalk--CUSUM has power 1.00 and conditional MAE 1.12 under the full structural alternative. Node2Vec--CUSUM is the closest random-walk embedding competitor, with power 0.83 and conditional MAE 2.99; LSE--CUSUM and ASE--CUSUM have powers 0.75 and 0.62. The low powers of the edge-count, RW-LS, and LAD statistics distinguish the tested community reallocation from a first-order density shift.

\subsection{Signal Strength and Theory-Guided Detection Behavior}

We first vary the magnitude of a density-preserving structural change. Let $Q_0$ and $Q_1$ denote the three- and five-community probability matrices in the main design. The pre-change matrix remains $Q_0$, while the post-change matrix is
\begin{equation}
Q_{\lambda}=(1-\lambda)Q_0+\lambda Q_1,
\qquad
\lambda\in\{0,0.25,0.50,0.75,1\}.
\label{eq:signal-interpolation}
\end{equation}
Because $Q_0$ and $Q_1$ have identical expected densities, every interpolated alternative is also density matched. The endpoint $\lambda=0$ is a no-change evaluation setting and $\lambda=1$ is the main structural alternative. Each method uses one critical value calibrated from 999 independent $Q_0$ null sequences, followed by 200 independent evaluation sequences at each value of $\lambda$.

\begin{figure*}[t]
\centering
\includegraphics[width=0.96\textwidth]{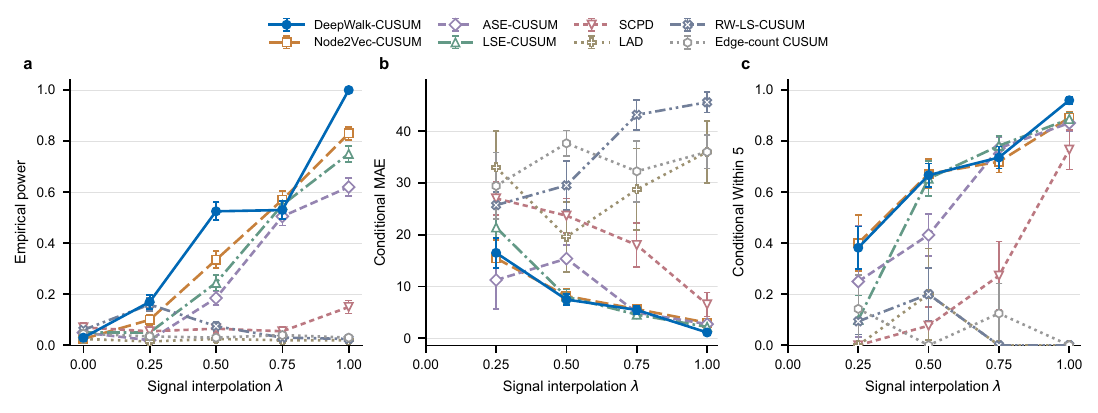}
\caption{Calibrated detection and localization across the density-preserving signal path in~\eqref{eq:signal-interpolation}. (a) Empirical rejection probability, equal to size at $\lambda=0$ and power for $\lambda>0$; (b) MAE conditional on rejection; and (c) conditional Within 5. Localization is undefined at $\lambda=0$ and is therefore omitted from panels (b)--(c). Error bars are Monte Carlo standard errors over $R=200$ evaluation sequences.}
\label{fig:signal-strength}
\end{figure*}

\begin{table*}[t]
\centering
\caption{Calibrated rejection probability across density-preserving signal levels. The $\lambda=0$ column is empirical size; the remaining columns are empirical power. Parentheses report Monte Carlo standard errors over $R=200$ evaluation sequences. Bold identifies the best result in each column.}
\label{tab:signal-strength}
\begingroup
\footnotesize
\setlength{\tabcolsep}{3.0pt}
\begin{tabular}{lccccc}
\toprule
Method & $\lambda=0$ & $\lambda=0.25$ & $\lambda=0.50$ & $\lambda=0.75$ & $\lambda=1$ \\
\midrule
DeepWalk--CUSUM & 0.03 (0.012) & \textbf{0.17 (0.027)} & \textbf{0.53 (0.035)} & 0.53 (0.035) & \textbf{1.00 (0.000)} \\
Node2Vec--CUSUM & 0.03 (0.011) & 0.10 (0.021) & 0.34 (0.033) & \textbf{0.57 (0.035)} & 0.83 (0.027) \\
ASE--CUSUM & \textbf{0.05 (0.015)} & 0.02 (0.010) & 0.18 (0.027) & 0.51 (0.035) & 0.62 (0.034) \\
LSE--CUSUM & \textbf{0.05 (0.015)} & 0.05 (0.015) & 0.24 (0.030) & 0.55 (0.035) & 0.75 (0.031) \\
SCPD & 0.07 (0.018) & 0.06 (0.016) & 0.07 (0.017) & 0.06 (0.016) & 0.15 (0.025) \\
LAD & 0.03 (0.011) & 0.01 (0.009) & 0.03 (0.011) & 0.02 (0.010) & 0.02 (0.010) \\
RW-LS--CUSUM & 0.06 (0.017) & 0.16 (0.026) & 0.07 (0.019) & 0.03 (0.012) & 0.03 (0.011) \\
Edge-count CUSUM & 0.04 (0.015) & 0.04 (0.013) & 0.03 (0.012) & 0.04 (0.014) & 0.03 (0.012) \\
\bottomrule
\end{tabular}
\endgroup
\end{table*}

Figure~\ref{fig:signal-strength} and Table~\ref{tab:signal-strength} separate detection from localization. DeepWalk--CUSUM increases from power 0.17 at $\lambda=0.25$ to 1.00 at $\lambda=1$; its conditional MAE decreases from 16.50 to 1.12 and its conditional Within-5 rate increases from 0.38 to 0.96. Node2Vec--CUSUM and LSE--CUSUM reach powers 0.83 and 0.75 at the full alternative. At $\lambda=0.75$, their powers are slightly above the DeepWalk value, with differences comparable to the reported Monte Carlo errors. RW-LS has some power only at the weakest nonzero signal, while the density-only statistic remains near its null rejection rate throughout. The calibrated experiment supports a detection-to-localization transition without treating a forced CUSUM maximizer as evidence of detection.

We next connect the simulations to the fluctuation scale in Theorem~\ref{thm:null_fluctuation}. Up to constants, define the diagnostic proxy
\begin{equation}
\widetilde R_T=
\frac{1}{\gamma_K}
\sqrt{\frac{\log n}{n\bar p}}
\sqrt{d+\log T},
\label{eq:theory-diagnostic-proxy}
\end{equation}
where $\bar p$ is expected edge density and $\gamma_K$ is the smaller pre- and post-change population gap. In the effective-degree experiment, both regime matrices are multiplied by a common factor to obtain expected degrees in $\{8,12,18,30\}$ while preserving their normalized block structure. In the gap experiment, each regime matrix is interpolated toward the constant-density matrix using multipliers $\{0.35,0.55,0.75,1\}$. This keeps density fixed but changes both the population gap and the structural contrast. We therefore interpret the latter as a difficulty path consistent with the theorem, not as a causal isolation of the gap. The three representation methods directly tied to the spectral argument, DeepWalk, ASE, and LSE, are evaluated with 499 independent calibration sequences and 100 alternatives per setting.

\begin{figure*}[t]
\centering
\includegraphics[width=0.96\textwidth]{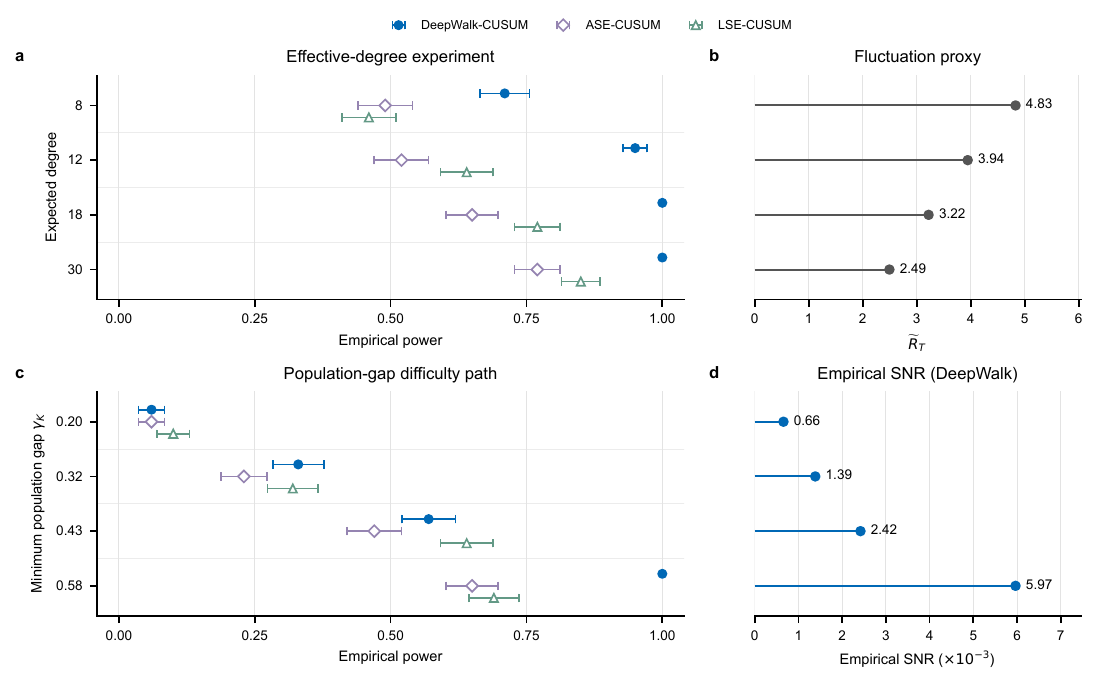}
\caption{Theory-guided detection diagnostics. (a) Empirical power at each expected-degree setting; (b) the corresponding fluctuation proxy $\widetilde R_T$ in~\eqref{eq:theory-diagnostic-proxy}; (c) empirical power at each setting of the constant-density population-gap difficulty path; and (d) the empirical signal-to-fluctuation ratio $\sqrt{k^*(T-k^*)/T}\,\widehat\Delta/\widetilde R_T$ for DeepWalk, displayed in units of $10^{-3}$. Points and horizontal error bars in (a) and (c) report empirical power plus or minus one Monte Carlo standard error over $R=100$ alternatives. Methods are offset vertically within each setting for readability. Zero-based line segments and endpoints in (b) and (d) display the diagnostic values.}
\label{fig:theory-guided}
\end{figure*}

Figure~\ref{fig:theory-guided}(a)--(b) shows the predicted benefit of greater effective connectivity. As expected degree increases from 8 to 30, $\widetilde R_T$ decreases from 4.83 to 2.49 and DeepWalk power rises from 0.71 to 1.00; ASE and LSE also improve. Along the gap path, the minimum population gap increases from 0.203 to 0.580, and DeepWalk power rises from 0.06 to 1.00. The accompanying empirical signal-to-fluctuation ratio also increases. These controlled trends support the signal--noise mechanism of the theory, while the joint change in gap and structural contrast limits a one-factor interpretation of the second diagnostic.

\subsection{Multiple Change-Point Detection}
For the multiple-change experiment, we take $T=300$ with true change set $\mathcal K^*=\{100,200\}$ and three density-matched regimes with three, five, and four communities. The number of changes is not supplied to the methods. Each method calibrates its stopping threshold as the higher empirical 95\% quantile of its maximum scan statistic over $B_{\rm cal}=499$ independent no-change sequences. For vector representations, binary segmentation recursively splits a segment only when its local CUSUM statistic exceeds this threshold, with minimum segment length 30. LAD and SCPD use their native scores with retrospective reference windows of 5 and 10 snapshots; threshold-exceeding local peaks are retained with minimum separation 30. Evaluation uses 100 additional no-change and 100 alternative sequences, with paired graph-sequence seeds across methods. For estimated set $\widehat{\mathcal K}$, accuracy is summarized by the Hausdorff distance
\begin{equation}
d_H(\widehat{\mathcal K},\mathcal K^*)=
\max\left\{
\max_{a\in\widehat{\mathcal K}}\min_{b\in\mathcal K^*}|a-b|,
\max_{b\in\mathcal K^*}\min_{a\in\widehat{\mathcal K}}|a-b|
\right\}.
\end{equation}
When $\widehat{\mathcal K}$ is empty under the alternative, we set $d_H=T$ for reporting.

\begin{figure*}[t]
\centering
\includegraphics[width=0.94\textwidth]{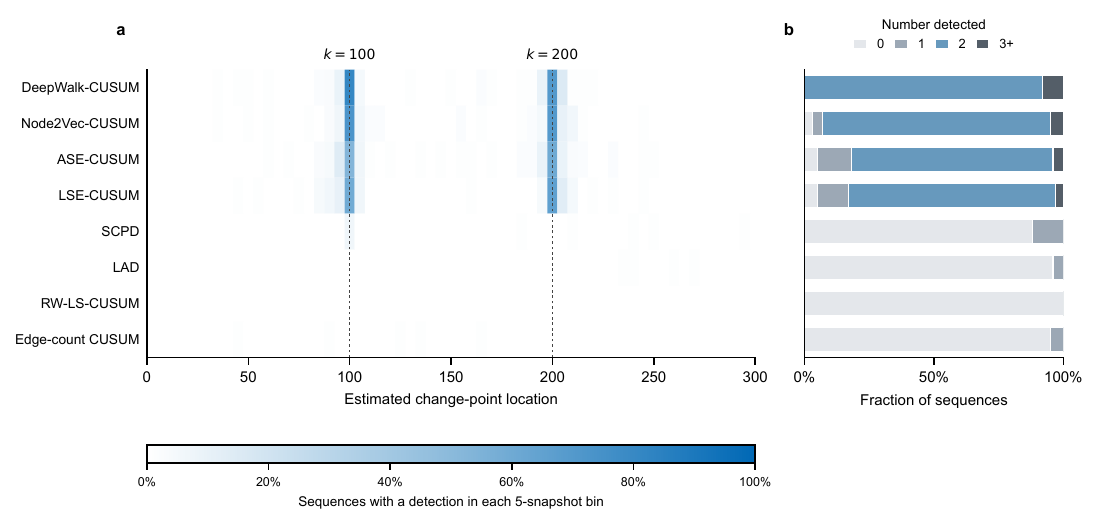}
\caption{Multiple-change detection with calibrated automatic stopping. (a) Fraction of $R=100$ alternative graph-sequence replications reporting at least one change in each five-snapshot bin. All methods share the same color scale, and the denominator includes replications with no detected change. Dashed vertical lines mark the true changes at 100 and 200. (b) Distribution of the number of estimated changes, grouped as 0, 1, 2, and 3 or more.}
\label{fig:multichange-autostop}
\end{figure*}

\begin{table*}[t]
\centering
\caption{Multiple-change detection with threshold-calibrated automatic stopping. The true change points are 100 and 200 in a length-300 sequence. FPR is the probability of reporting at least one change under no change; power is the probability of reporting at least one change under the alternative. Results use 499 calibration, 100 no-change, and 100 alternative replications. Best values are bold; for FPR and mean $\widehat m$, best means closest to 0.05 and 2, respectively. Ties are retained.}
\label{tab:multiple-change-comparison}
\begingroup
\footnotesize
\setlength{\tabcolsep}{3.5pt}
\begin{tabular}{lccccccc}
\toprule
Method & FPR & Power & Exact count & All within 5 & Mean $d_H$ & Median $d_H$ & Mean $\widehat m$ \\
\midrule
DeepWalk--CUSUM & 0.01 & \textbf{1.00} & \textbf{0.92} & \textbf{0.77} & \textbf{6.41} & \textbf{2.0} & 2.08 \\
Node2Vec--CUSUM & \textbf{0.05} & 0.97 & 0.88 & 0.69 & 18.74 & \textbf{2.0} & \textbf{1.96} \\
ASE--CUSUM & 0.03 & 0.95 & 0.78 & 0.55 & 33.61 & 4.0 & 1.81 \\
LSE--CUSUM & 0.00 & 0.95 & 0.80 & 0.59 & 31.97 & 4.0 & 1.81 \\
SCPD & 0.04 & 0.12 & 0.00 & 0.00 & 277.85 & 300.0 & 0.12 \\
LAD & 0.06 & 0.04 & 0.00 & 0.00 & 294.07 & 300.0 & 0.04 \\
RW-LS--CUSUM & 0.04 & 0.00 & 0.00 & 0.00 & 300.00 & 300.0 & 0.00 \\
Edge-count CUSUM & 0.03 & 0.05 & 0.00 & 0.00 & 289.71 & 300.0 & 0.05 \\
\bottomrule
\end{tabular}
\endgroup
\end{table*}

Figure~\ref{fig:multichange-autostop} and Table~\ref{tab:multiple-change-comparison} show that DeepWalk--CUSUM detects at least one change in all alternative sequences, estimates exactly two changes in 92\% of replications, and places both changes within five snapshots in 77\% of replications. Its mean and median Hausdorff errors are 6.41 and 2.0. Node2Vec--CUSUM has power 0.97, exact-count rate 0.88, all-within-5 rate 0.69, and mean and median Hausdorff errors 18.74 and 2.0. ASE--CUSUM and LSE--CUSUM both have power 0.95, with all-within-5 rates 0.55 and 0.59. DeepWalk is conservative under no change, with FPR 0.01, compared with 0.05 for Node2Vec. The FPRs of RW-LS, LAD, and SCPD range from 0.04 to 0.06, while their powers range from 0.00 to 0.12. None recovers both changes in these 100 replications. Edge-count CUSUM likewise has little power in this density-matched construction. These are finite-sample comparisons; the replication-level outcomes and Monte Carlo standard errors are included in the accompanying results.

\subsection{Robustness and Ablation Studies}

An intentionally misspecified node-order perturbation diagnostic is reported in Appendix~\ref{app:node-order-perturbation}. It delineates the failure induced by withholding correspondence information and is not evidence of robustness to unmatched node populations.

We also evaluate whether the fixed-reference Procrustes alignment step is necessary for embedding-based CUSUM statistics. The ablation is run under the main density-matched single-change design with $R=100$ replications and compares graph-level vectors obtained by mean pooling before and after alignment.

\begin{table}[t]
\centering
\caption{Paired alignment ablation under the main density-matched design ($R=100$). The two variants share the same networks and fitted embeddings within each method and replication. Parentheses report Monte Carlo standard errors; the best result within each alignment variant is bold.}
\label{tab:alignment-ablation}
\begingroup
\scriptsize
\setlength{\tabcolsep}{1.2pt}
\begin{tabular}{@{}lcccc@{}}
\toprule
& \multicolumn{2}{c}{Unaligned} & \multicolumn{2}{c}{Procrustes aligned} \\
\cmidrule(lr){2-3}\cmidrule(lr){4-5}
Method & MAE & Within 5 & MAE & Within 5 \\
\midrule
DeepWalk & \textbf{2.44 (0.40)} & \textbf{0.87 (0.034)} & \textbf{1.12 (0.16)} & \textbf{0.97 (0.017)} \\
Node2Vec & 19.56 (1.63) & 0.31 (0.046) & 5.20 (0.90) & 0.77 (0.042) \\
ASE & 33.06 (1.34) & 0.07 (0.026) & 9.97 (1.40) & 0.59 (0.049) \\
LSE & 31.66 (1.51) & 0.08 (0.027) & 7.85 (1.26) & 0.67 (0.047) \\
\bottomrule
\end{tabular}
\endgroup
\end{table}

Table~\ref{tab:alignment-ablation} confirms that coordinate alignment is an important stabilization step for independently fitted embeddings. Procrustes alignment reduces the MAE from 2.44 to 1.12 for DeepWalk--CUSUM and increases its Within-5 rate from 0.87 to 0.97. The gains are larger for the other embedding methods: the MAE decreases from 19.56 to 5.20 for Node2Vec--CUSUM, from 33.06 to 9.97 for ASE--CUSUM, and from 31.66 to 7.85 for LSE--CUSUM. These paired comparisons isolate the effect of coordinate alignment from variation in the simulated networks and fitted embeddings.

Because the theory allows a fixed or independently estimated reference, we also compare three prespecified reference choices for DeepWalk--CUSUM: the first snapshot $X_1$, the medoid of the first ten fitted embeddings under pairwise Procrustes distance, and an embedding from an independently generated null graph. The three variants share the same graph sequence and fitted snapshot embeddings within each replication. Each variant is calibrated separately using 999 null sequences and evaluated on 200 alternatives.

\begin{table}[t]
\centering
\caption{Reference-choice sensitivity for DeepWalk--CUSUM. Localization is conditional on rejection; parentheses report Monte Carlo standard errors. Best values are bold, with ties retained.}
\label{tab:reference-choice}
\begingroup
\scriptsize
\setlength{\tabcolsep}{2.0pt}
\begin{tabular}{@{}lccc@{}}
\toprule
Reference & Power & Cond. MAE & Cond. Within 5 \\
\midrule
First snapshot $X_1$ & \textbf{1.000 (0.000)} & \textbf{0.915 (0.119)} & 0.965 (0.013) \\
Early-window medoid & \textbf{1.000 (0.000)} & 0.950 (0.142) & \textbf{0.970 (0.012)} \\
Independent null & \textbf{1.000 (0.000)} & 0.950 (0.112) & \textbf{0.970 (0.012)} \\
\bottomrule
\end{tabular}
\endgroup
\end{table}

Table~\ref{tab:reference-choice} shows no material sensitivity within this prespecified reference set: all three choices have power 1.00, conditional MAE between 0.915 and 0.950, and conditional Within-5 rates between 0.965 and 0.970. This result supports the fixed-reference implementation used elsewhere, but it does not remove the requirement for known correspondence or informative common anchors.

Finally, we varied the change location over $k^*/T\in\{0.25,0.50,0.75\}$ under the full density-matched alternative. DeepWalk--CUSUM has power at least 0.99 at all three locations, with conditional MAE between 1.09 and 1.52 and conditional Within-5 rates between 0.95 and 0.96. Node2Vec--CUSUM has power between 0.73 and 0.84, while ASE--CUSUM and LSE--CUSUM are less stable near the sequence boundaries. Appendix~\ref{app:change-location} reports the complete comparison and makes explicit that localization is conditional on rejection.

We additionally considered a sparse degree-corrected SBM to assess robustness beyond homogeneous node degrees. LSE--CUSUM attained the lowest MAE, 7.06, in this setting. DeepWalk--CUSUM attained an MAE of 10.53 and a Within-5 rate of 0.56, compared with 10.73 and 0.54 for ASE--CUSUM. Node2Vec--CUSUM was less stable, with MAE 19.31 and Within 5 equal to 0.33. DeepWalk also outperformed RW-LS, LAD, SCPD, and edge-count CUSUM. Appendix~\ref{app:dcsbm-stress} gives the complete design and results. This experiment identifies degree normalization as beneficial under strong degree heterogeneity, while showing that DeepWalk remains competitive with an unnormalized spectral embedding.

We also varied one DeepWalk hyperparameter at a time using paired graph sequences; Appendix~\ref{app:hyperparameter-sensitivity} reports the complete results. Performance was stable for $d\in\{4,8,16\}$, 5--20 walks per node, and walk lengths of 15--60. Their Within-5 rates ranged from 0.92 to 0.97. The context window had a larger effect: window sizes 5, 10, and 15 produced MAEs of 0.23, 1.29, and 3.23. We retained the prespecified window size of 10 in all primary comparisons rather than selecting it after observing this diagnostic.

\subsection{Computational Scalability}

We next examine the tradeoff among localization accuracy, representation dimension, and end-to-end computation as the number of nodes increases. We use the same density-matched design with $n\in\{100,150,200,300\}$ and $R=30$ paired replications at each size. As specified in Section~V-A, balanced community sizes and the exact finite-$n$ density formula are used at every $n$. Runtime includes representation construction and the downstream scan and is measured serially using one Python process and one BLAS thread. These measurements characterize the present implementations and hardware; they are not universal complexity constants.

\begin{figure*}[t]
\centering
\includegraphics[width=0.96\textwidth]{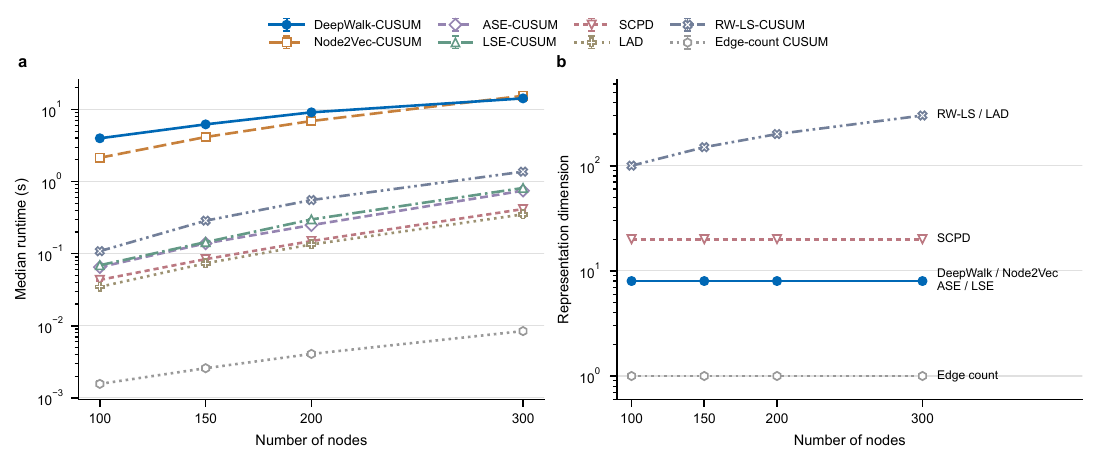}
\caption{Computational scalability and representation dimension. (a) Median serial end-to-end runtime over $R=30$ paired graph sequences, with interquartile error bars. (b) Dimension of each graph representation or spectral descriptor; methods with identical dimensions share a labeled curve. Both vertical axes use logarithmic scales. DeepWalk, Node2Vec, ASE, and LSE use $d=8$ graph-level vectors, SCPD uses 20 bins, RW-LS and LAD use $n$-dimensional descriptors, and edge-count CUSUM has dimension one.}
\label{fig:scalability-dimension}
\end{figure*}

Figure~\ref{fig:scalability-dimension} shows that the median DeepWalk runtime increases from 3.98 seconds at $n=100$ to 14.22 seconds at $n=300$; Node2Vec increases from 2.14 to 15.44 seconds. ASE and LSE are substantially faster in this implementation, requiring 0.75 and 0.81 seconds at $n=300$. The same paired runs show that DeepWalk MAE decreases from 2.97 to 0.07 and Node2Vec MAE decreases from 5.37 to 0.60 as $n$ increases. At $n=300$, ASE and LSE have MAEs of 2.87 and 1.80, respectively. These embedding comparisons indicate a localization-accuracy--runtime tradeoff under this design, not a general speed advantage. Fixed-dimensional summaries keep the downstream CUSUM input compact, but this property is shared by DeepWalk, Node2Vec, ASE, and LSE and does not imply faster representation construction. Appendix~\ref{app:runtime-details} reports medians and interquartile intervals for every method and network size.

The density-matched single-change, calibrated signal, theory-guided, automatic multi-change, node-order perturbation, reference-choice, change-location, scalability, alignment-ablation, sparse degree-corrected SBM, and hyperparameter-sensitivity experiments were reproduced from archived scripts using the fixed-reference alignment in~\eqref{eq:procrustes}. Node2Vec is included in the main single-change, calibrated signal, size/power, multiple-change, change-location, and runtime comparisons. Seeds, software versions, full hyperparameters, and artifact-level provenance are documented in the replication output and Appendix~\ref{app:simulation-reproducibility}.

\section{Real Data Analysis}
\label{sec:real}

\subsection{Data and Network Construction}

International trade systems can naturally be represented as dynamic weighted networks, where countries or regions are connected through trade relationships and trade volumes correspond to edge weights. As global supply chains, international trade policies, and geopolitical environments evolve over time, the underlying network structure may also exhibit substantial temporal variation.

To illustrate the structural features of international trade networks, Figure~\ref{fig:trade_network} (a) provides a visualization of a representative trade network. Nodes denote countries or regions, edges represent trade relationships, and edge weights reflect trade volumes. Different colors indicate modularity classes identified from the network structure. The visualization illustrates the heterogeneous distribution of export relationships and trade intensities across trading partners.

\begin{figure}[!t]
\centering
\includegraphics[width=\columnwidth]{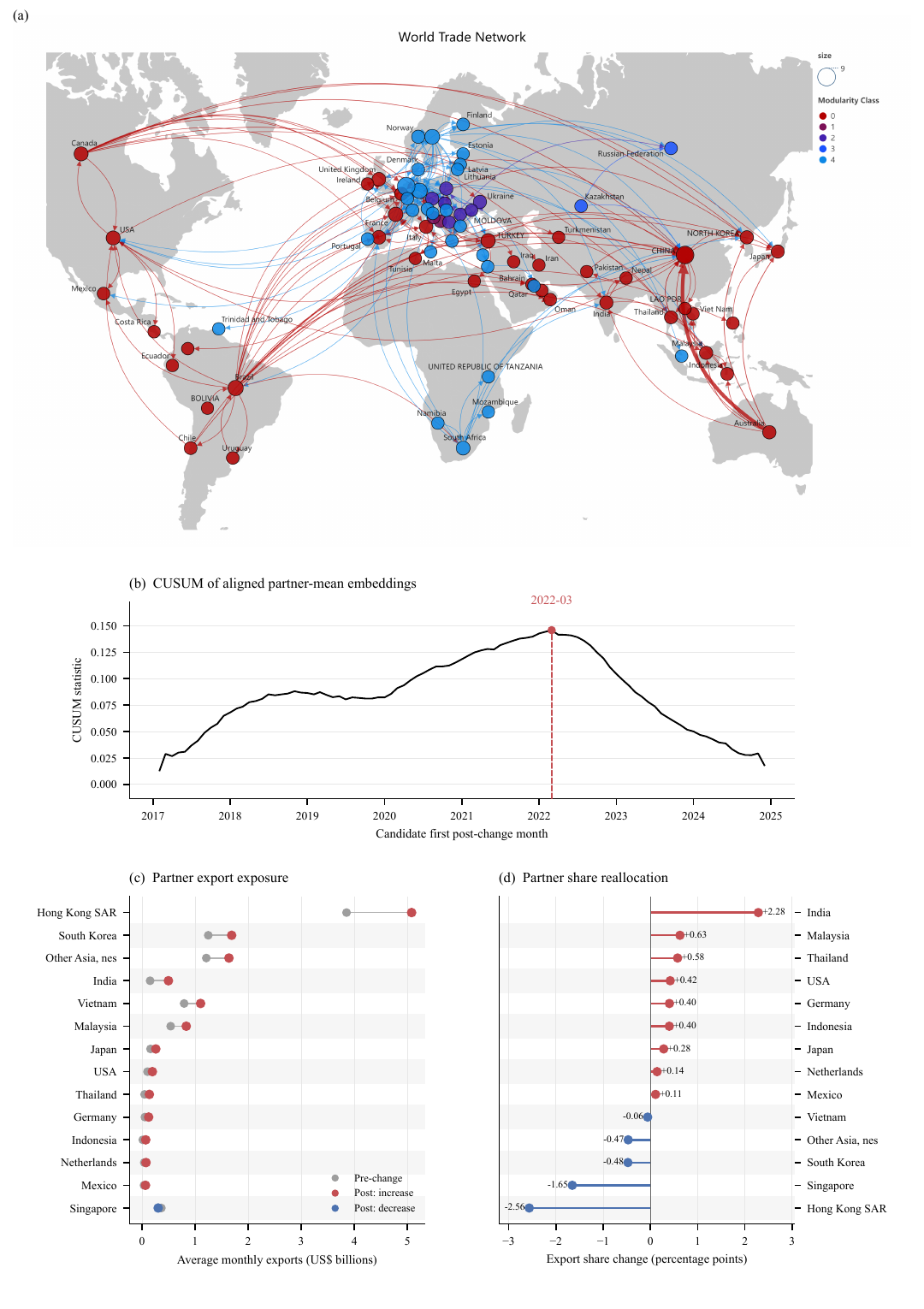}
\caption{China's HS8542 export-network analysis. Panel (a) visualizes the international trade network structure. Panel (b) shows the CUSUM statistic of the aligned monthly graph representations, with the dashed line indicating March 2022 as the first post-change month. Panel (c) compares average monthly exports before and after the detected change point. Panel (d) reports the corresponding changes in export shares in percentage points. Panels (c) and (d) display the 14 trading partners with the largest absolute changes in average monthly exports.}
\label{fig:trade_network}
\end{figure}

Motivated by this dynamic-network perspective, we study structural changes in China's electronic integrated circuit export network using monthly trade data from the United Nations Comtrade Database (UN Comtrade)~\citep{uncomtrade}. We consider China's exports under HS8542 from January 2017 to December 2024. After removing non-specific partner categories and invalid records, the resulting data contain 11,262 positive-export observations involving 204 trading partners over 96 months. Each month is represented as a weighted trade network, where reported export flows define links between China and its trading partners, while unreported partner links are treated as absent.

For each month \(t\), we construct a weighted star centered on China. Its leaves are the active trading partners \(V_t^+\), defined by strictly positive reported export values. Node identities are fixed across the sample using Comtrade partner codes; inactive partners do not participate in that month's embedding. For random-walk generation, we use an undirected adjacency with \(w_{\mathrm{China},i,t}=w_{i,\mathrm{China},t}=x_{i,t}\), where \(x_{i,t}\) is China's HS8542 export value to partner \(i\). The reverse adjacency permits walks to return to China and does not add import values. Thus, the economic relationships are exports from China, while the embedding operator uses their symmetrized adjacency.

We apply weighted DeepWalk independently to each monthly network. For each active node \(i\), random walks are generated according to the weighted transition probability
\(P_t(i\to j)=w_{ij,t}/\sum_k w_{ik,t}\).
We generate 50 random walks per node, each with a length of 10, and use a maximum context window of 5. The resulting walk sequences are then used to learn node embeddings through the Skip-Gram model with negative sampling. We set the embedding dimension to \(d=16\), use five negative samples and five training epochs, and fix the random seed at 42 for reproducibility.

To ensure comparability of independently fitted embeddings across months, we use January 2017 as the fixed reference. Let \(V_t^+\) denote the active trading partners in month \(t\), excluding China, and define the common anchors by
\(\mathcal A_t=V_t^+\cap V_{\rm ref}^+\).
The current and reference anchor matrices are arranged in the same partner-code order. Given

$$
X_{t,\mathcal A_t}^{\top}X_{{\rm ref},\mathcal A_t}
=
U_t\Sigma_tV_t^{\top},
$$

we set \(\widehat Q_t=U_tV_t^{\top}\) and align all active-node embeddings through
\(\widetilde X_t=X_t\widehat Q_t\).
The reference month uses the identity transformation.

Because the set of active trading partners varies across months, we use the active-node analogue of the mean-pooling representation in~\eqref{eq:graph-embedding},

$$
Z_t
=
\frac{1}{|V_t^+|}
\sum_{i\in V_t^+}\widetilde X_{t,i}
\in\mathbb R^{16}.
$$

This produces 96 monthly graph representations in a common coordinate system. The number of common anchors ranges from 84 to 109 across months, substantially exceeding the embedding dimension \(d=16\).

Because the random-walk transition probabilities depend on normalized edge weights, a uniform rescaling of all export values within a snapshot does not change the resulting walk distribution. Accordingly, the detected changes should be interpreted primarily as shifts in export allocation and partner participation rather than changes in aggregate export volume.

\subsection{Empirical Results}

Figure~\ref{fig:trade_network}(b) presents the CUSUM statistic computed from the aligned monthly graph representations. The maximum occurs at \(\widehat{k}=62\), corresponding to a split between February and March 2022. We therefore identify March 2022 as the first month of the estimated post-change regime.

The estimated change location is highly stable across the examined embedding configurations. All combinations of \(d\in\{8,16,32\}\) and random seeds \(\{11,22,33,42,55\}\) identify the same March 2022 transition. The estimated location also remains unchanged under alternative reference months, further supporting the stability of the detected transition.

Panels~\ref{fig:trade_network}(c) and (d) provide a descriptive view of the structural changes around the estimated break. Panel~(c) compares average monthly exports before and after March 2022 for the trading partners exhibiting the largest absolute changes, whereas panel~(d) reports the corresponding changes in export shares. The two panels reveal substantial heterogeneity across trading partners and show that changes in absolute export exposure do not necessarily coincide with changes in relative export importance.

The estimated change point in March 2022 closely follows changes in the international export-control environment in late February 2022. On February 24--25, 2022, the United States and the European Union introduced additional export restrictions involving semiconductors, technology products, and dual-use goods destined for Russia~\citep{bisRussia2022,euRussia2022}. March 2022 is the first full calendar month following these measures. Together with the partner-level changes shown in Fig.~\ref{fig:trade_network}(c)--(d), this temporal alignment is consistent with a broader reorganization of China's HS8542 export allocation during this period, although the analysis is not intended for causal attribution.

Overall, the empirical analysis shows that the proposed aligned DeepWalk--CUSUM framework identifies a stable and interpretable structural transition in China's HS8542 export network.

\section{Conclusion}
\label{sec:conclusion}

This paper proposes a DeepWalk embedding-based method for change-point detection in dynamic networks.
The proposed method transforms graph-valued observations into a sequence of low-dimensional vectors through random-walk embeddings, uses fixed-reference orthogonal Procrustes alignment to make independently fitted embeddings comparable, and uses a multivariate CUSUM statistic for testing and localization.
The theoretical analysis shows that the population co-occurrence operator of DeepWalk and the graph transition operator share the same dominant spectral subspace.
Under the small initialization and spectral gap conditions, the empirical embeddings concentrate around the population spectral structure.
Furthermore, we derive a stochastic fluctuation bound under the null hypothesis, a detection boundary under the alternative hypothesis, and a localization error bound.
These results explicitly characterize the relationships among the number of nodes, network sparsity, spectral gap, embedding dimension, and signal strength.
The calibrated simulations show that detection improves with structural signal and effective connectivity, and the theory-guided diagnostics exhibit the predicted association with the fluctuation scale. The multiple-change, reference-choice, change-location, degree-heterogeneity, and runtime experiments further identify where DeepWalk--CUSUM is accurate and where normalized spectral representations or simpler summaries are computationally preferable. The trade network analysis illustrates the method on a weighted empirical sequence.

This study still has several limitations.
First, the theoretical analysis mainly focuses on independent snapshots, a fixed or independently estimated reference embedding, and changes in the embedding mean.
When the reference is estimated from the analyzed sequence, its induced temporal dependence requires additional analysis. Procrustes alignment also requires known node correspondence or sufficiently informative common anchors.
Temporal dependence, changes in the node set, and more general distributional changes require further investigation.
Second, mean pooling may overlook certain forms of local heterogeneity, and future work may consider learnable permutation-invariant aggregation operators.
Finally, the multiple change-point procedure adopts classical binary segmentation.
More complex settings involving narrow or densely spaced change points may be addressed by extending the method with wild binary segmentation or dynamic programming approaches.

\appendices
\section{Sparse Degree-Corrected SBM Stress Experiment}
\label{app:dcsbm-stress}

The main simulations use homogeneous block-specific connection probabilities. To examine a more heterogeneous and sparse setting, we generated networks from a degree-corrected stochastic block model (DCSBM). We used $n=150$, $T=100$, a single change point at $k^*=50$, and $R=100$ independent graph-sequence replications. The pre-change and post-change networks contained three and five equal-sized communities, respectively. Node identities were fixed over time.

For each replication, node propensities were independently generated from a log-normal distribution with log-scale standard deviation $0.6$. They were truncated to $[0.4,2.0]$ and rescaled to have sample mean one. The same propensity vector $\theta=(\theta_1,\ldots,\theta_n)^\top$ was retained across all snapshots in that replication. Conditional on the community labels $g_{i,t}$ and propensities, the upper-triangular adjacency entries were independently generated according to
\begin{equation}
\Pr(A_{ij,t}=1\mid \theta,g_t)
=\theta_i\theta_j
\begin{cases}
p_{\rm in}^{(r)}, & g_{i,t}=g_{j,t},\\
p_{\rm out}, & g_{i,t}\ne g_{j,t},
\end{cases}
\label{eq:dcsbm-edge-probability}
\end{equation}
where $r=0$ before the change and $r=1$ after the change. We set $p_{\rm in}^{(0)}=0.12$ and $p_{\rm out}=0.02$.

The post-change within-community probability was selected separately in each replication to match the propensity-weighted expected edge density. Specifically, let
\begin{equation}
W_r=\sum_{i<j}\theta_i\theta_j
\mathbf{1}\{g_i^{(r)}=g_j^{(r)}\},
\qquad
W=\sum_{i<j}\theta_i\theta_j.
\end{equation}
Equating the expected numbers of edges before and after the change gives
\begin{equation}
p_{\rm in}^{(1)}
=p_{\rm out}+\frac{W_0}{W_1}
\left(p_{\rm in}^{(0)}-p_{\rm out}\right).
\label{eq:dcsbm-density-matching}
\end{equation}
Across the 100 replications, the mean value of $p_{\rm in}^{(1)}$ was 0.1889, and the mean coefficient of variation of the node propensities was 0.508. The realized average degrees before and after the change were 7.872 and 7.870. The corresponding fractions of isolated nodes were 0.0073 and 0.0070. Thus, the experiment introduced substantial degree heterogeneity and sparsity without adding a first-order density signal.

DeepWalk used $d=8$, ten walks of length 30 per node, a context-window size of 10, and fixed-reference Procrustes alignment. Every method received the same graph sequence within each replication. Table~\ref{tab:dcsbm-stress} reports localization performance; the values in parentheses are Monte Carlo standard errors.

\begin{table}[H]
\centering
\caption{Localization under the sparse degree-corrected SBM stress experiment. Smaller MAE and median absolute error and larger Within 5 are preferred; the best value in each column is bold. Results are based on $R=100$ paired graph-sequence replications.}
\label{tab:dcsbm-stress}
\begingroup
\scriptsize
\setlength{\tabcolsep}{2.5pt}
\begin{tabular}{lccc}
\toprule
Method & MAE & Median AE & Within 5 \\
\midrule
LSE--CUSUM        &  \textbf{7.06 (1.08)} &  \textbf{2.5} & \textbf{0.65 (0.048)} \\
DeepWalk--CUSUM   & 10.53 (1.37) &  5.0 & 0.56 (0.050) \\
ASE--CUSUM        & 10.73 (1.35) &  5.0 & 0.54 (0.050) \\
RW-LS--CUSUM      & 16.57 (1.81) &  6.5 & 0.47 (0.050) \\
Node2Vec--CUSUM   & 19.31 (1.67) & 18.0 & 0.33 (0.047) \\
SCPD              & 22.52 (1.46) & 22.0 & 0.17 (0.038) \\
LAD               & 22.82 (1.30) & 23.5 & 0.13 (0.034) \\
Edge-count CUSUM  & 33.60 (1.45) & 37.0 & 0.05 (0.022) \\
\bottomrule
\end{tabular}
\endgroup
\end{table}

LSE--CUSUM achieved the lowest MAE and the largest Within-5 rate in this setting. DeepWalk--CUSUM and ASE--CUSUM had similar localization accuracy, with differences smaller than their Monte Carlo standard errors. Node2Vec--CUSUM was less stable than both, with MAE 19.31 and Within 5 equal to 0.33. DeepWalk--CUSUM was also more accurate than RW-LS--CUSUM and the native LAD and SCPD scores. Edge-count CUSUM remained inaccurate because the expected edge density was matched across the change. The improvement of LSE--CUSUM over the unnormalized embedding methods is consistent with the explicit degree normalization in the Laplacian representation.

\section{DeepWalk Hyperparameter Sensitivity}
\label{app:hyperparameter-sensitivity}

We evaluated DeepWalk--CUSUM under the main density-matched single-change design with $n=150$, $T=100$, $k^*=50$, and $R=100$ replications. The default configuration used embedding dimension $d=8$, ten walks per node, walk length 30, and context-window size 10. We varied one component at a time over
\begin{equation}
\begin{split}
d&\in\{4,8,16\},\qquad M\in\{5,10,20\},\\
L&\in\{15,30,60\},\qquad w\in\{5,10,15\},
\end{split}
\label{eq:hyperparameter-grid}
\end{equation}
where $M$, $L$, and $w$ denote the number of walks per node, walk length, and context-window size. All configurations within a replication received the same graph sequence and independent random-walk seeds. The default row was recomputed within this paired experiment and therefore uses a different Monte Carlo sample from Table~\ref{tab:density-matched-single}.

\begin{table}[H]
\centering
\caption{One-factor-at-a-time sensitivity of DeepWalk--CUSUM. The default is $(d,M,L,w)=(8,10,30,10)$. Parentheses report Monte Carlo standard errors over $R=100$ paired graph-sequence replications. The smallest errors and largest Within 5 are bold; ties are retained.}
\label{tab:hyperparameter-sensitivity}
\begingroup
\scriptsize
\setlength{\tabcolsep}{2.2pt}
\begin{tabular}{lcccc}
\toprule
Setting & Value & MAE & Median AE & Within 5 \\
\midrule
Default          & -- & 1.29 (0.20) & \textbf{0.0} & 0.95 (0.022) \\
Dimension $d$    &  4 & 1.43 (0.22) & 0.5 & 0.95 (0.022) \\
Dimension $d$    & 16 & 0.97 (0.18) & \textbf{0.0} & 0.97 (0.017) \\
Walk count $M$   &  5 & 1.43 (0.24) & 1.0 & 0.93 (0.026) \\
Walk count $M$   & 20 & 0.81 (0.16) & \textbf{0.0} & 0.97 (0.017) \\
Walk length $L$  & 15 & 0.83 (0.15) & \textbf{0.0} & 0.97 (0.017) \\
Walk length $L$  & 60 & 1.32 (0.23) & \textbf{0.0} & 0.92 (0.027) \\
Window size $w$  &  5 & \textbf{0.23 (0.06)} & \textbf{0.0} & \textbf{1.00 (0.000)} \\
Window size $w$  & 15 & 3.23 (0.61) & 1.0 & 0.83 (0.038) \\
\bottomrule
\end{tabular}
\endgroup
\end{table}

Table~\ref{tab:hyperparameter-sensitivity} shows limited sensitivity to the embedding dimension. Increasing the number of walks from 5 to 20 modestly improved localization, whereas longer walks did not improve accuracy in this design. The context-window size produced the clearest difference. A window of 5 attained perfect Within-5 localization, while a window of 15 reduced that rate to 0.83. These results support the stability of the default configuration across most components but also identify the context window as a parameter that should be reported and assessed in new applications.

\section{Node-Order Perturbation Diagnostic}
\label{app:node-order-perturbation}

As a boundary-of-applicability diagnostic, we independently permuted the rows and columns at each time point and deliberately withheld the permutation maps from every method. This construction is intentionally misspecified relative to the fixed-correspondence model used by DeepWalk--CUSUM. With known identities, all methods could first restore a common ordering; with genuinely unknown identities, the proposed Procrustes step has no reliable anchors. The experiment therefore measures sensitivity to unavailable correspondence information, not robustness to node relabeling.

\begin{table*}[!t]
\centering
\caption{Node-order perturbation diagnostic with deliberately unavailable correspondence. In the single-change setting, larger power and Within 5 and smaller localization errors are preferred. In the pure perturbation null, smaller FPR indicates less sensitivity to arbitrary matrix ordering. The diagnostic uses its original split-sample 95\% critical values calibrated from 100 independent no-change sequences and is not directly compared with the formal calibration in Table~\ref{tab:size-power}. Dashes indicate that localization metrics are not applicable because no true change point exists under the null. Results are based on $R=100$ replications.}
\label{tab:relabeling-robustness}
\begingroup
\footnotesize
\setlength{\tabcolsep}{5pt}
\begin{tabular}{llcccc}
\toprule
Setting & Method & Power/FPR & MAE & Median AE & Within 5 \\
\midrule
\multirow{4}{*}{Random relabeling with one change}
& DeepWalk--CUSUM & \textbf{1.00} & \textbf{2.81} & \textbf{1.0} & \textbf{0.84} \\
& SCPD & 0.10 & 14.43 & 10.0 & 0.43 \\
& LAD & 0.02 & 20.67 & 19.5 & 0.21 \\
& RW-LS--CUSUM & 0.02 & 38.87 & 44.0 & 0.01 \\
\cmidrule(lr){1-6}
\multirow{4}{*}{Pure relabeling null}
& DeepWalk--CUSUM & 0.08 & -- & -- & -- \\
& SCPD & \textbf{0.00} & -- & -- & -- \\
& LAD & \textbf{0.00} & -- & -- & -- \\
& RW-LS--CUSUM & 0.12 & -- & -- & -- \\
\bottomrule
\end{tabular}
\endgroup
\end{table*}

DeepWalk--CUSUM has power 1.00, MAE 2.81, and FPR 0.08 in this misspecified diagnostic, whereas the permutation-invariant LAD and SCPD scores do not reject under the pure perturbation null but have low power under the structural alternative. The result is retained only to make the correspondence requirement explicit.

\section{Additional Change-Location Results}
\label{app:change-location}

The main experiments place the change at the sequence midpoint. To examine the interior-location condition in Theorem~\ref{thm:detection_boundary}, we repeated the full density-matched alternative at $k^*/T\in\{0.25,0.50,0.75\}$. Each method uses the critical value obtained from its 999 independent null calibration sequences, and each location uses $R=200$ independent alternatives. Power is the rejection probability. MAE and Within 5 are computed only among rejected sequences, so a method with low power can have an apparently moderate conditional localization value based on very few detections.

\begin{table*}[!t]
\centering
\caption{Change-location sensitivity. Each cell reports power / conditional MAE / conditional Within 5 over $R=200$ alternatives. Bold identifies the best result for each location. A dash indicates that conditional localization is undefined because the method never rejected.}
\label{tab:change-location}
\begingroup
\footnotesize
\setlength{\tabcolsep}{4.0pt}
\begin{tabular}{lccc}
\toprule
Method & $k^*/T=0.25$ & $k^*/T=0.50$ & $k^*/T=0.75$ \\
\midrule
DeepWalk--CUSUM & \textbf{0.99/1.24/0.96} & \textbf{0.99/1.09/0.96} & \textbf{0.99/1.52/0.95} \\
Node2Vec--CUSUM & 0.73/3.84/0.84 & 0.84/2.65/0.88 & 0.74/2.34/0.86 \\
ASE--CUSUM & 0.57/4.35/0.83 & 0.72/3.14/0.85 & 0.59/4.70/0.79 \\
LSE--CUSUM & 0.64/2.49/0.90 & 0.75/2.87/0.81 & 0.65/4.42/0.83 \\
SCPD & 0.14/8.75/0.79 & 0.16/10.44/0.62 & 0.15/9.07/0.83 \\
LAD & 0.02/37.00/0.25 & 0.02/27.25/0.00 & 0.04/39.86/0.00 \\
RW-LS--CUSUM & 0.01/21.00/0.00 & 0.00/--/-- & 0.01/55.00/0.00 \\
Edge-count CUSUM & 0.02/23.75/0.00 & 0.03/25.00/0.17 & 0.03/49.33/0.00 \\
\bottomrule
\end{tabular}
\endgroup
\end{table*}

DeepWalk--CUSUM changes little across the three tested locations despite the smaller CUSUM multiplier away from the midpoint. The spectral embedding comparators lose more power near one or both boundaries, while the native score and density-only methods rarely reject. These conclusions are limited to changes that remain well inside the sequence; behavior closer to the endpoints is not assessed.

\section{Runtime Details}
\label{app:runtime-details}

Table~\ref{tab:runtime-scalability} gives the numerical values underlying Fig.~\ref{fig:scalability-dimension}. Each timing is the serial wall-clock time for representation construction and the downstream scan on one length-100 sequence. Interquartile intervals summarize run-to-run variation across 30 paired sequences. The measurements are implementation- and hardware-specific and should be interpreted together with the representation dimensions in Fig.~\ref{fig:scalability-dimension}(b).

\begin{table*}[t]
\centering
\caption{Serial end-to-end runtime in seconds for a length-100 graph sequence. Entries are medians with interquartile intervals in parentheses over $R=30$ paired sequences. The smallest median at each network size is bold. Subsecond timings retain additional decimal places to avoid rounding nonzero runtimes to zero.}
\label{tab:runtime-scalability}
\begingroup
\footnotesize
\setlength{\tabcolsep}{4pt}
\begin{tabular}{@{}lcccc@{}}
\toprule
Method & $n=100$ & $n=150$ & $n=200$ & $n=300$ \\
\midrule
DeepWalk--CUSUM & 3.98 (3.93, 4.02) & 6.21 (6.18, 6.28) & 9.11 (9.07, 9.15) & 14.22 (13.82, 14.53) \\
Node2Vec--CUSUM & 2.14 (2.11, 2.15) & 4.14 (4.11, 4.16) & 6.90 (6.85, 6.97) & 15.44 (15.22, 15.98) \\
ASE--CUSUM & 0.065 (0.064, 0.067) & 0.139 (0.137, 0.142) & 0.249 (0.246, 0.259) & 0.747 (0.722, 0.767) \\
LSE--CUSUM & 0.069 (0.068, 0.071) & 0.144 (0.142, 0.146) & 0.298 (0.263, 0.308) & 0.814 (0.784, 0.838) \\
SCPD & 0.043 (0.042, 0.044) & 0.084 (0.082, 0.085) & 0.149 (0.142, 0.164) & 0.415 (0.380, 0.443) \\
LAD & 0.035 (0.034, 0.035) & 0.074 (0.072, 0.075) & 0.134 (0.131, 0.141) & 0.350 (0.322, 0.401) \\
RW-LS--CUSUM & 0.108 (0.106, 0.110) & 0.286 (0.282, 0.292) & 0.554 (0.541, 0.566) & 1.37 (1.33, 1.40) \\
Edge-count CUSUM & \textbf{0.0016 (0.0015, 0.0017)} & \textbf{0.0026 (0.0024, 0.0028)} & \textbf{0.0041 (0.0039, 0.0042)} & \textbf{0.0084 (0.0076, 0.0088)} \\
\bottomrule
\end{tabular}
\endgroup
\end{table*}

\section{Simulation Reproducibility}
\label{app:simulation-reproducibility}

The accompanying replication package contains the simulation scripts, fixed base seeds and deterministic method-specific seed rules, row-level Monte Carlo outputs, machine-readable run metadata, and the code used to generate every simulation table and figure. Calibration and evaluation sequences are stored as distinct phases, and paired comparisons reuse graph-sequence seeds while assigning separate random-number streams to stochastic embedding methods. A consolidated manifest maps each reported simulation artifact to its generating script, output directory, protocol, and validation checks. Runtime measurements additionally record the one-process, one-BLAS-thread execution rule used for Fig.~\ref{fig:scalability-dimension} and Table~\ref{tab:runtime-scalability}.

\section{Proofs of the Main Results}
\label{app:proofs}

\subsection{Proof of Spectral structure of the DeepWalk operator}
\label{app:proof-1}

\begin{proof}
Since
\[
\widetilde P
=
U\Lambda U^\top,
\]
we have
\[
\widetilde P^t
=
U\Lambda^tU^\top
\]
for every integer \(t\ge1\).

Therefore,
\[
\mathcal M
=
2\sum_{t=1}^{T_w}(L-t)
D_A^{1/2}
U\Lambda^tU^\top
D_A^{-1/2}\Pi.
\]

Observe that
\[
\Lambda^t
=
\mathrm{diag}(\lambda_1^t,\dots,\lambda_n^t).
\]
Hence, \(\mathcal M\) is obtained by applying a polynomial spectral transformation to the
transition operator \(P\).

Since polynomial transformations preserve eigenspaces, the eigenvectors of
\(\mathcal M\) coincide with those of \(P\). In particular, the leading spectral subspace
\[
\mathcal U_K
=
\mathrm{span}(u_1,\dots,u_K)
\]
is invariant under \(\mathcal M\).

This completes the proof.
\end{proof}

\subsection{Proof of Linearized DeepWalk dynamics}
\label{app:proof-2}

\begin{proof}
The proof proceeds by linearizing the gradient of the DeepWalk objective around the
small initialization regime.

Recall that
\[
\mathcal L(X,Y)
=
-
\sum_{i,j}
C_{ij}
\log(Q_{ij}),
\]
with
\[
Q_{ij}
=
\frac{
e^{\langle X_i,Y_j\rangle}
}{
\sum_{k=1}^n e^{\langle X_i,Y_k\rangle}
}.
\]

Define
\[
z_{ij}
=
\langle X_i,Y_j\rangle.
\]
Under Condition (A2) ,
\[
|z_{ij}|
=
o(1).
\]

Using the first-order Taylor expansion of the exponential function,
\[
e^{z_{ij}}
=
1+z_{ij}+O(z_{ij}^2).
\]

Therefore,
\[
\sum_{k=1}^n e^{z_{ik}}
=
n+\sum_{k=1}^n z_{ik}
+
O\left(
\sum_{k=1}^n z_{ik}^2
\right).
\]

Applying the expansion
\[
\frac1{n+a}
=
\frac1n
-
\frac{a}{n^2}
+
O(a^2),
\]
we obtain
\[
Q_{ij}
=
\frac1n
+
\frac1n
\left(
z_{ij}
-
\frac1n\sum_{k=1}^n z_{ik}
\right)
+
O(\|W\|_F^2).
\]

Substituting this expansion into the gradient of the objective function yields
\[
\nabla_W\mathcal L(W)
=
\mathcal L_{\mathcal M}W
+
O(\|W\|_F^2),
\]
where
\[
\mathcal L_{\mathcal M}
\]
is a linear operator depending on the co-occurrence structure encoded in
\[
\mathcal M.
\]

Consequently, the gradient descent iteration becomes
\[
W^{(s+1)}
=
W^{(s)}
-
\eta
\mathcal L_{\mathcal M}W^{(s)}
+
R^{(s)},
\]
where
\[
\|R^{(s)}\|_F
=
O(\|W^{(s)}\|_F^2).
\]

Absorbing the linear term into the operator
\[
\mathcal L_{\mathcal M},
\]
we obtain
\[
W^{(s+1)}
=
\mathcal L_{\mathcal M}W^{(s)}
+
R^{(s)}.
\]

This completes the proof.
\end{proof}

\subsection{Proof of Embedding concentration}
\label{app:proof-3}

\begin{proof}
Since the empirical co-occurrence matrix satisfies
\[
C=\mathcal M+E,
\]
the empirical DeepWalk operator can be viewed as a perturbation of the population
operator.

By Proposition~1, the leading eigenspace of
\[
\mathcal M
\]
coincides with the leading eigenspace of the graph transition operator.

Applying the Davis--Kahan sin-theta theorem yields
\[
\|\sin\Theta(\widehat U_K,U_K)\|
\le
\frac{
\|E\|
}{
\gamma_K
}.
\]

Since
\[
\|E\|
=
O_P\left(
\sqrt{
\frac{\log n}{n\rho}
}
\right),
\]
we conclude that
\[
\|\sin\Theta(\widehat U_K,U_K)\|
=
O_P\left(
\frac1{\gamma_K}
\sqrt{
\frac{\log n}{n\rho}
}
\right).
\]

This completes the proof.
\end{proof}

\subsection{Proof of Null fluctuation bound for DeepWalk-CUSUM}
\label{app:proof-4}

\begin{proof}
Under the null hypothesis,
\[
Z_t=\mu+\varepsilon_t.
\]
Substituting this decomposition into the CUSUM statistic yields
\[
S_T(k)
=
\sqrt{
\frac{k(T-k)}{T}
}
\left(
\frac1k
\sum_{t=1}^k \varepsilon_t
-
\frac1{T-k}
\sum_{t=k+1}^T \varepsilon_t
\right).
\]

Define the partial sum process
\[
R_k
=
\sum_{t=1}^k \varepsilon_t.
\]

Since
\[
\sum_{t=k+1}^T\varepsilon_t
=
R_T-R_k,
\]
we obtain
\[
S_T(k)
=
\sqrt{
\frac{T}{k(T-k)}
}
\left(
R_k-\frac{k}{T}R_T
\right).
\]

Fix
\[
u\in\mathbb S^{d-1}.
\]
Then
\[
\langle u,S_T(k)\rangle
=
\sqrt{
\frac{T}{k(T-k)}
}
\left(
\langle u,R_k\rangle
-
\frac{k}{T}\langle u,R_T\rangle
\right).
\]

Observe that
\[
\langle u,S_T(k)\rangle
\]
is a linear combination of independent sub-Gaussian random variables. Hence,
\[
\langle u,S_T(k)\rangle
\]
is itself sub-Gaussian with variance proxy proportional to
\[
\sigma_\varepsilon^2.
\]

More precisely, there exists a universal constant \(C>0\) such that
\[
\mathbb P
\left(
|\langle u,S_T(k)\rangle|>x
\right)
\le
2\exp\left(
-\frac{x^2}{C\sigma_\varepsilon^2}
\right).
\]

To control the Euclidean norm, let
\[
\mathcal N
\]
be a \(1/2\)-net of the unit sphere
\[
\mathbb S^{d-1}
\]
satisfying
\[
|\mathcal N|
\le
5^d.
\]

Using the standard covering argument,
\[
\|S_T(k)\|
\le
2
\max_{u\in\mathcal N}
|\langle u,S_T(k)\rangle|.
\]

Therefore,
\[
\mathbb P
\left(
\|S_T(k)\|>2x
\right)
\le
2\cdot 5^d
\exp\left(
-\frac{x^2}{C\sigma_\varepsilon^2}
\right).
\]

Applying the union bound over
\[
k=1,\dots,T-1,
\]
we obtain
\[
\mathbb P
\left(
W_T>2x
\right)
\le
2T5^d
\exp\left(
-\frac{x^2}{C\sigma_\varepsilon^2}
\right).
\]

Choosing
\[
x
=
C_1
\sigma_\varepsilon
\sqrt{
d+\log T
},
\]
for a sufficiently large constant
\[
C_1>0,
\]
yields
\[
W_T
=
O_P\left(
\sigma_\varepsilon
\sqrt{
d+\log T
}
\right).
\]

Finally, by Theorem~3,
\[
\sigma_\varepsilon
=
O\left(
\frac1{\gamma_K}
\sqrt{
\frac{\log n}{n\rho}
}
\right).
\]

Substituting this bound gives
\[
W_T
=
O_P\left(
\frac1{\gamma_K}
\sqrt{
\frac{\log n}{n\rho}
}
\sqrt{
d+\log T
}
\right).
\]

When
\[
d
\]
is fixed,
\[
\sqrt{
d+\log T
}
\asymp
\sqrt{\log T},
\]
which implies
\[
W_T
=
O_P\left(
\frac1{\gamma_K}
\sqrt{
\frac{\log n\log T}{n\rho}
}
\right).
\]

This completes the proof.
\end{proof}

\subsection{Proof of Signal structure of the DeepWalk-CUSUM process}
\label{app:proof-5}

\begin{proof}
We first evaluate the deterministic CUSUM process at the true change point.

Since
\[
\mu_t=\mu^{(1)},
\qquad
1\le t\le k^*,
\]
and
\[
\mu_t=\mu^{(2)},
\qquad
k^*<t\le T,
\]
we obtain
\[
\frac1{k^*}
\sum_{t=1}^{k^*}\mu_t
=
\mu^{(1)},
\]
and
\[
\frac1{T-k^*}
\sum_{t=k^*+1}^T\mu_t
=
\mu^{(2)}.
\]

Therefore,
\[
S_T^{(\mu)}(k^*)
=
\sqrt{
\frac{k^*(T-k^*)}{T}
}
\left(
\mu^{(1)}-\mu^{(2)}
\right).
\]

Taking norms yields
\[
\|S_T^{(\mu)}(k^*)\|
=
\sqrt{
\frac{k^*(T-k^*)}{T}
}
\Delta.
\]

Next, for any
\[
k\neq k^*,
\]
the averaging windows mix observations from different regimes, which reduces the signal
magnitude relative to the true change point. A direct calculation shows that the signal
curve is piecewise linear in
\[
|k-k^*|.
\]

Under Condition (A4),
\[
\frac{k^*}{T}
\]
remains bounded away from the boundary, implying that the local slope of the signal curve
is uniformly bounded below by a positive constant depending only on
\[
\tau_0.
\]

Consequently, there exists
\[
c_0>0
\]
such that
\[
\|S_T^{(\mu)}(k^*)\|
-
\|S_T^{(\mu)}(k)\|
\ge
c_0
\Delta
\frac{|k-k^*|}{\sqrt T}.
\]

This completes the proof.
\end{proof}

\subsection{Proof of Detection boundary for DeepWalk-CUSUM}
\label{app:proof-6}

\begin{proof}
Since
\[
W_T
=
\max_{1\le k<T}
\|S_T(k)\|,
\]
we have
\[
W_T
\ge
\|S_T(k^*)\|.
\]

Using the decomposition
\[
S_T(k^*)
=
S_T^{(\mu)}(k^*)
+
S_T^{(\varepsilon)}(k^*),
\]
the reverse triangle inequality yields
\[
\|S_T(k^*)\|
\ge
\|S_T^{(\mu)}(k^*)\|
-
\|S_T^{(\varepsilon)}(k^*)\|.
\]

By Lemma~\ref{lem:signal_structure},
\[
\|S_T^{(\mu)}(k^*)\|
=
\sqrt{
\frac{k^*(T-k^*)}{T}
}
\Delta.
\]

Moreover, by Theorem~\ref{thm:null_fluctuation},
\[
\max_{1\le k<T}
\|S_T^{(\varepsilon)}(k)\|
=
O_P\left(
\frac1{\gamma_K}
\sqrt{
\frac{\log n}{n\rho}
}
\sqrt{
d+\log T
}
\right).
\]

Hence,
\[
\|S_T^{(\varepsilon)}(k^*)\|
=
O_P\left(
\frac1{\gamma_K}
\sqrt{
\frac{\log n}{n\rho}
}
\sqrt{
d+\log T
}
\right).
\]

Combining the above inequalities yields
\[
W_T
\ge
\sqrt{
\frac{k^*(T-k^*)}{T}
}
\Delta
-
O_P\left(
\frac1{\gamma_K}
\sqrt{
\frac{\log n}{n\rho}
}
\sqrt{
d+\log T
}
\right).
\]

Finally, if
\[
\sqrt{
\frac{k^*(T-k^*)}{T}
}
\Delta
\gg
\frac1{\gamma_K}
\sqrt{
\frac{\log n}{n\rho}
}
\sqrt{
d+\log T
},
\]
the deterministic signal dominates the stochastic fluctuation term, implying
\[
W_T
\overset P\longrightarrow
\infty.
\]

This completes the proof.
\end{proof}

\subsection{Proof of Localization consistency}
\label{app:proof-7}

\begin{proof}
Since
\[
\widehat k
=
\arg\max_{1\le k<T}
\|S_T(k)\|,
\]
we have
\[
\|S_T(\widehat k)\|
\ge
\|S_T(k^*)\|.
\]

Using the decomposition
\[
S_T(k)
=
S_T^{(\mu)}(k)
+
S_T^{(\varepsilon)}(k),
\]
together with the triangle inequality and reverse triangle inequality, yields
\[
\|S_T^{(\mu)}(k^*)\|
-
\|S_T^{(\mu)}(\widehat k)\|
\le
2
\max_{1\le k<T}
\|S_T^{(\varepsilon)}(k)\|.
\]

By Lemma~1,
\[
\|S_T^{(\mu)}(k^*)\|
-
\|S_T^{(\mu)}(\widehat k)\|
\ge
c_0
\Delta
\frac{
|\widehat k-k^*|
}{\sqrt T}.
\]

Combining the above inequalities yields
\[
c_0
\Delta
\frac{
|\widehat k-k^*|
}{\sqrt T}
\le
2
\max_{1\le k<T}
\|S_T^{(\varepsilon)}(k)\|.
\]

Applying Theorem~\ref{thm:null_fluctuation} gives
\[
\max_{1\le k<T}
\|S_T^{(\varepsilon)}(k)\|
=
O_P\left(
\frac1{\gamma_K}
\sqrt{
\frac{\log n}{n\rho}
}
\sqrt{
d+\log T
}
\right).
\]

Therefore,
\[
|\widehat k-k^*|
=
O_P\left(
\frac{
\sqrt T
}{
\gamma_K\Delta
}
\sqrt{
\frac{\log n}{n\rho}
}
\sqrt{
d+\log T
}
\right).
\]

Finally, if
\[
\sqrt T\Delta
\gg
\frac1{\gamma_K}
\sqrt{
\frac{\log n}{n\rho}
}
\sqrt{
d+\log T
},
\]
then
\[
\frac{
|\widehat k-k^*|
}{T}
=
o_P(1).
\]

This completes the proof.
\end{proof}

\bibliographystyle{IEEEtranN}
\bibliography{Reference}

\end{document}